\documentclass{article}

\usepackage[english]{babel}
\usepackage{latexsym}
\usepackage{amssymb}
\usepackage{amsmath}
\usepackage{amsthm}
\usepackage{cite}    
\usepackage{ifthen}
\usepackage{xspace}
\usepackage{tikz}
\usepackage{paralist}
\usepackage{mathdots}

\usepackage[colorlinks,final,citecolor=blue]{hyperref}

\usetikzlibrary{%
  arrows,%
  positioning,%
  decorations.pathmorphing,%
  decorations.pathreplacing,%
}

\xdefinecolor{uwopurple}{cmyk}{.86,1,0,.12}

\colorlet{colAux}{blue!20!white}
\colorlet{colMod}{orange!80!black}
\colorlet{colModd}{uwopurple}
\colorlet{colPos}{green!60!black}
\colorlet{colNeg}{red!75!black}
\colorlet{colUni}{yellow}

\colorlet{eins}{colPos}
\colorlet{zwei}{white}
\colorlet{drei}{colAux}

\theoremstyle{plain}
\newtheorem{theorem}{Theorem}[]
\newtheorem{proposition}[theorem]{Proposition}
\newtheorem{lemma}[theorem]{Lemma}
\newtheorem{corollary}[theorem]{Corollary}

\theoremstyle{definition}

\newtheorem*{prob}{Problem}

\theoremstyle{remark}
\newtheorem{remark}[theorem]{Remark}

\newcommand{\myproblem}[3]{\pagebreak[2]%
	\begin{prob}[#3]\ \\ \vspace{-\baselineskip}
		\begin{compactdesc}
			\item{\sc Given:} #1;
			\item{\sc Output:} #2.
		\end{compactdesc}
	\end{prob}}

\newcommand{\cI}{\ensuremath{\mathcal{I}}\xspace}

\renewcommand{\phi}{\varphi}

\newcommand{\sse}{\subseteq}

\newcommand{\sm}{\setminus}

\newcommand{\ie}{i.\,e.,\xspace}
\newcommand{\eg}{e.\,g.,\xspace}
\newcommand{\resp}{resp.,\ }

\newcommand{\sett}[2]{\left\{#1\mathrel{\left|\vphantom{#1}\vphantom{#2}\right.}#2\right\}}
\newcommand{\set}[1]{\left\{\mathinner{#1}\right\}}

\newcommand{\abs}[1]{\left|\mathinner{#1}\right|}

\newcommand{\gen}[1]{\langle \mathinner{#1} \rangle}

\newcommand{\N}{\mathbb{N}}

\newcommand{\NP}{\ensuremath{\mathrm{NP}}\xspace}

\newcommand\bul{\ensuremath{{\bullet}}\xspace}
\newcommand\dia{\ensuremath{{\blacklozenge}}\xspace}
\newcommand\str{\ensuremath{{\bigstar}}\xspace}
\newcommand\tri{\ensuremath{{\blacktriangle}}\xspace}

\newcommand\tas{\mbox{DRTAS}\xspace}
\newcommand\pats{{\sc \mbox{Pats}}\xspace}
\newcommand\kpats{{\sc \mbox{$k$-Pats}}\xspace}
\newcommand\mbpats{{\sc \mbox{mbPats}}\xspace}
\newcommand\tmbpats{{\sc \mbox{$3$-mbPats}}\xspace}
\newcommand\kmbpats{{\sc \mbox{$k$-mbPats}}\xspace}
\newcommand\sat{{\sc \mbox{$3$-Sat}}\xspace}
\newcommand\modpats{{\sc Modified} \pats}

\newcommand{\smalltile}[4]
{
	\begin{scope}[xshift=#1 cm,yshift=#2 cm]
		\draw [black, fill=#3] (-.225,-.225) rectangle (.225,.225);
		\node [font=\footnotesize] at (0,0) {#4};
	\end{scope}
}

\newcommand\colT[2]{%
	\mbox{\hspace{.225cm}\tikz%
	[baseline=-.6ex,overlay,scale=.85,every node/.style={scale=1}]
	{\smalltile{0}{0}{#1}{#2}}\hspace{.225cm}}\xspace}
\newcommand\Cor{\colT{white}{\tl{or}}}
\newcommand\Cra{\colT{white}{$\overset0\rightarrow$}}
\newcommand\Cpos{\colT{colPos}{\color{white}$\pmb+$}}
\newcommand\Cneg{\colT{colNeg}{\color{white}$\pmb-$}}

\newcommand{\fulltile}[9]
{
	\begin{scope}[xshift=#1 cm,yshift=#2 cm]
		\draw [black, fill=#3] (-.7,-.7) rectangle (.7,.7);
		\begin{scope}[#4,inner sep=2pt]
			\node [font=\small] at (0,0) {$#5$};
			\node [font=\footnotesize,anchor=north] at (0,.7) {$#6$};
			\node [font=\footnotesize,anchor=north,rotate=-90] at (.7,0) {$#7$};
			\node [font=\footnotesize,anchor=south] at (0,-.7) {$#8$};
			\node [font=\footnotesize,anchor=south,rotate=-90] at (-.7,0) {$#9$};
		\end{scope}
	\end{scope}
}

\newcommand{\supertile}[4]
{
	\begin{scope}[xshift=#1 cm,yshift=#2 cm]
		\fill [black] (-.85,-.85) rectangle (.85,.85);
		\fill [lightgray] (-.85,.85) rectangle (1,1);
		\fill [lightgray] (.85,-.85) rectangle (1,1);
		\fill [white] (.85,.15*#3-.85) rectangle (1,.15*#3-.7);
		\fill [white] (.15*#3-.85,.85) rectangle (.15*#3-.7,1);
		\draw [black] (-1,-1) rectangle (1,1);
		\node [white,font=\small,text width=1.6cm,align=center] {\sf #4};
	\end{scope}
}

\newcommand{\supertileL}[4]
{
	\begin{scope}[xshift=#1 cm,yshift=#2 cm]
		\fill [black] (-.85,-.85) rectangle (.85,.85);
		\fill [lightgray] (-.85,.85) rectangle (1,1);
		\fill [lightgray] (.85,-.85) rectangle (1,1);
		\fill [white] (.85,.15*#3-.85) rectangle (1,.15*#3-.7);
		\fill [white] (.15*#3-.85,.85) rectangle (.15*#3-.7,1);
		\draw [black] (-.85,-1) rectangle (1,1);
		\node [white,font=\small,text width=1.6cm,align=center] {\sf #4};
	\end{scope}
}

\newcommand{\supertileB}[4]
{
	\begin{scope}[xshift=#1 cm,yshift=#2 cm]
		\fill [black] (-.85,-.85) rectangle (.85,.85);
		\fill [lightgray] (-.85,.85) rectangle (1,1);
		\fill [lightgray] (.85,-.85) rectangle (1,1);
		\fill [white] (.85,.15*#3-.85) rectangle (1,.15*#3-.7);
		\fill [white] (.15*#3-.85,.85) rectangle (.15*#3-.7,1);
		\draw [black] (-1,-.85) rectangle (1,1);
		\node [white,font=\small,text width=1.6cm,align=center] {\sf #4};
	\end{scope}
}

\newcommand{\supertileBL}[4]
{
	\begin{scope}[xshift=#1 cm,yshift=#2 cm]
		\fill [black] (-.85,-.85) rectangle (.85,.85);
		\fill [lightgray] (-.85,.85) rectangle (1,1);
		\fill [lightgray] (.85,-.85) rectangle (1,1);
		\fill [white] (.85,.15*#3-.85) rectangle (1,.15*#3-.7);
		\fill [white] (.15*#3-.85,.85) rectangle (.15*#3-.7,1);
		\draw [black] (-.85,-.85) rectangle (1,1);
		\node [white,font=\small,text width=1.6cm,align=center] {\sf #4};
	\end{scope}
}

\newcommand\tl[1]{\ensuremath{\mathsf{#1}}}

\def\smallscale{1}
\def\fullscale{1}
\def\fullscalenode{1}

\begin{document}

\title{3-color Bounded Patterned Self-assembly%
\texorpdfstring{%
	\thanks{
		The research of L.~K. and S.~K. was supported
		by the NSERC Discovery Grant R2824A01 and
		UWO Faculty of Science grant to L.~K.
		The research of S.~S. was supported
		by the HIIT Pump Priming Project Grant 902184/T30606.}}{}}
		
\author{Lila Kari$^1$ \and Steffen Kopecki$^1$ \and 
	Shinnosuke Seki$^2$}
	
\date{}

\maketitle

{\small \centering
	$^1$ Department of Computer Science, Middlesex College, \\
	The University of Western Ontario \\
	London Ontario N6A 5B7, Canada \\
	{\tt lila@csd.uwo.ca, steffen@csd.uwo.ca}	\\
	\medskip
	$^2$ Helsinki Institute of Information Technology (HIIT), \\
	Department of Computer Science, Aalto University \\
	P.\,O.~Box 15400, FI-00076, Aalto, Finland \\
	{\tt shinnosuke.seki@aalto.fi} \\
}

\begin{abstract}
Patterned self-assembly tile set synthesis (\pats) is the problem of finding a minimal tile set which uniquely self-assembles into a given pattern.
Czeizler and Popa proved the \NP-completeness of \pats and Seki showed that the \pats problem is already \NP-complete for patterns with 60 colors.
In search for the minimal number of colors such that \pats remains \NP-complete, we introduce multiple bound \pats (\mbpats) where we allow bounds for the numbers of tile types of each color.
We show that \mbpats is \NP-complete for patterns with just three colors and, as a byproduct of this result, we also obtain a novel proof for the \NP-completeness of \pats which is more concise than the previous proofs.
\end{abstract}

\section{Introduction}\label{sec:intro}

Tile self-assembly is the autonomous formation of a structure from individual {\em tiles} controlled by local attachment rules.
One application of self-assembly is the implementation of nanoscopic tiles by DNA strands forming double crossover tiles with four unbounded single strands \cite{WinfreeLWS1998}.
The unbounded single strands control the assembly of the structure as two, or more, tiles can attach to each other only if the bonding strength between these single strands is big enough.
The general concept is to have many copies of the same tile types in a solution  which then form a large crystal-like structure over time; often an initial structure, the {\em seed}, is present in the solution from which the assembly process starts.

A mathematical model describing self-assembly systems is the {\em abstract tile self-assembly model} (aTAM), introduced by Winfree \cite{WinfreePhD}.
Many variants of aTAMs have been studied: a main distinction between the variants is whether the {\em shape} or the {\em pattern} of a self-assembled structure is studied.
In this paper we focus on the self-assembly of patterns, where a property, modeled as color, is assigned to each tile; see for example \cite{RothemundPW2004} where fluorescently labeled DNA tiles self-assemble into Sierpinski triangles.
Formally, a pattern is a rectilinear grid where each vertex has a color: a $k$-colored $m\times n$-pattern $P$ can be seen as a function $P\colon [m]\times [n] \to [k]$, where $[i] = \set{1,2,\ldots,i}$.
The optimization problem of {\em patterned self-assembly tile set synthesis} (\pats), introduced by Ma and Lombardi \cite{MaL2008}, is to determine the minimal number of tile types needed to uniquely self-assemble a given pattern starting from an $L$-shaped seed.
In this paper, we consider the decision variant of \pats, defined as follows:

\myproblem{A $k$-colored pattern $P$ and an integer $m$}
	{``Yes'' if $P$ can uniquely be self-assembled by using $m$ tile types}
	{\kpats}

Czeizler and Popa proved that \pats, where the number of colors on an input pattern is not bounded, is \NP-hard \cite{CzeizlerP2012}, but the practical interest lies in \kpats. Seki proved $60$-\pats is \NP-hard \cite{shin}.
By the nature of the biological implementations, the number of distinct colors in a pattern can be considered small.
In search for the minimal number $k$ for which \kpats remains \NP-hard, we investigate a modification of \pats: {\em multiple bound \pats} (\mbpats) uses individual bounds for the number of tile types of each color.

\myproblem{A pattern $P$ with colors from $[k]$ and $m_1,\ldots,m_k\in\N$}
	{``Yes'' if $P$ can uniquely be self-assembled by using $m_i$ tile types of color $i$, for $i\in[k]$}
	{\kmbpats}

The main contribution of this paper is a polynomial-time reduction from \pats to \tmbpats which proves the \NP-hardness of \tmbpats.
However, our reduction does not take every pattern as input, we only consider a restricted subset of patterns for which \pats is known to remain \NP-hard.
The patterns we use as input are exactly those patterns that are generated by a polynomial-time reduction from \sat to \pats.
Using one of the reductions which were presented in \cite{CzeizlerP2012,shin} as a foundation for our main result turned out to be unfeasible.
Therefore, we present a novel proof for the \NP-hardness of \pats which serves well as foundation for our main result.
Furthermore, our reduction from \sat to \pats is more concise compared to previous reductions in the sense that in order to self-assemble a pattern $P$  we only allow three more tile types than colors in $P$.
In Czeizler and Popa's approach the number of additional tile types is linear in the size of the input formula and Seki uses 84 tile types with 60 colors.

Let us note first that the decision variants of \pats and \mbpats can be solved in \NP by simple ``guess and check'' algorithms.
Before we prove \NP-hardness of \pats, Corollary~\ref{cor:pats} in Sect.~\ref{sec:pats}, and \tmbpats, Corollary~\ref{cor:mbpats} in Sect.~\ref{sec:bwg}, we introduce the formal concepts of patterned tile assembly systems, in Sect.~\ref{sec:RTAS}.

\section{Rectilinear Tile Assembly Systems}
\label{sec:RTAS}

In this section we formally introduce patterns and rectilinear tile assembly systems.
An excellent introduction to the fundamental model aTAM is given in \cite{RothemundW2000}.

Let $C$ be a finite {\em alphabet of colors}.
An $m\times n$-{\em pattern} $P$, for $m,n\in\N$, with colors from $C$ is a mapping $P\colon [m]\times[n] \to C$.
By $C(P) \sse C$ we denote the colors in the pattern $P$, \ie the codomain or range of the function $P$.
The pattern $P$ is called {\em $k$-colored} if $\abs{C(P)} \le k$.
The width and height of $P$ are denoted by $w(P) = m$ and $h(P) = n$, respectively.
We call $(x,y)\in[m]\times [n]$ a {\em position} in $P$.
The pattern is arranged such that position $(1,1)$ is on the bottom left, $(m,1)$ is on the bottom right, $(1,n)$ is on the top left, and $(m,n)$ is on the top right of the pattern $P$.
Fig.~\ref{fig:ex:pattern} (left side) shows an example pattern.

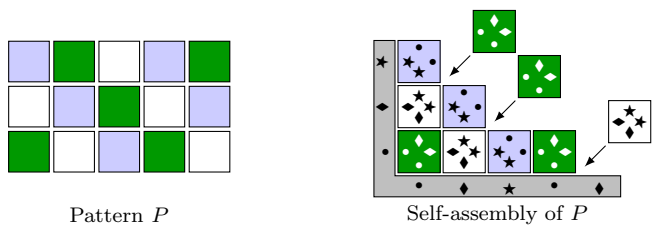
\begin{figure}[th]
	\vphantom x\hfill
	\begin{tikzpicture}[baseline=0,scale=1.2]
		\smalltile{0}{0}{eins}{}
		\smalltile{0}{.5}{zwei}{}
		\smalltile{0}{1}{drei}{}
		\smalltile{.5}{0}{zwei}{}
		\smalltile{.5}{.5}{drei}{}
		\smalltile{.5}{1}{eins}{}
		\smalltile{1}{0}{drei}{}
		\smalltile{1}{.5}{eins}{}
		\smalltile{1}{1}{zwei}{}
		\smalltile{1.5}{0}{eins}{}
		\smalltile{1.5}{.5}{zwei}{}
		\smalltile{1.5}{1}{drei}{}
		\smalltile{2}{0}{zwei}{}
		\smalltile{2}{.5}{drei}{}
		\smalltile{2}{1}{eins}{}
		\node [font=\footnotesize] at (1,-.7) {Pattern $P$};
	\end{tikzpicture}
	\hfill
	\begin{tikzpicture}[baseline=0]
	\begin{scope}[scale=2/5,every node/.style={scale=.7},
			font=\footnotesize]
		\draw [fill=lightgray] (-.8,-.8) -- (-.8,3.7) -- (-1.5, 3.7)
			-- (-1.5,-1.5) -- (6.7,-1.5) -- (6.7,-.8) -- cycle;
		\node [anchor=north] at (0,-.8) {\bul};
		\node [anchor=north] at (1.5,-.8) {\dia};
		\node [anchor=north] at (3,-.8) {\str};
		\node [anchor=north] at (4.5,-.8) {\bul};
		\node [anchor=north] at (6,-.8) {\dia};
		\node [anchor=north,rotate=-90] at (-.8,0) {\bul};
		\node [anchor=north,rotate=-90] at (-.8,1.5) {\dia};
		\node [anchor=north,rotate=-90] at (-.8,3) {\str};
		\fulltile{0}{0}{eins}{white}{}{\dia}{\dia}{\bul}{\bul}
		\fulltile{0}{1.5}{zwei}{black}{}{\str}{\str}{\dia}{\dia}
		\fulltile{0}{3}{drei}{black}{}{\bul}{\bul}{\str}{\str}
		\fulltile{1.5}{0}{zwei}{black}{}{\str}{\str}{\dia}{\dia}
		\fulltile{1.5}{1.5}{drei}{black}{}{\bul}{\bul}{\str}{\str}
		\fulltile{3}{0}{drei}{black}{}{\bul}{\bul}{\str}{\str}
		\fulltile{4.5}{0}{eins}{white}{}{\dia}{\dia}{\bul}{\bul}

		\fulltile{2.5}{4}{eins}{white}{}{\dia}{\dia}{\bul}{\bul}
		\fulltile{4}{2.5}{eins}{white}{}{\dia}{\dia}{\bul}{\bul}
		\fulltile{7}{1}{zwei}{black}{}{\str}{\str}{\dia}{\dia}
		\draw [-latex] (1.7,3.2) -- (1,2.5);
		\draw [-latex] (3.2,1.7) -- (2.5,1);
		\draw [-latex] (6.2,.2) -- (5.5,-.5);
	\end{scope}
		\node [font=\footnotesize] at (1.04,-.84) {Self-assembly of $P$};
	\end{tikzpicture}
	\hfill\vphantom x
	\caption{Pattern $P$ and how it can be self-assembled by three tile types.}
	\label{fig:ex:pattern}
\end{figure}

Let $\Sigma$ be a finite {\em alphabet of glues}.
A {\em colored Wang tile}, or simply {\em tile}, $t\in C\times \Sigma^4$ is a unit square with a color from $C$ and four glues from $\Sigma$, one on each of its edges.
The color of $t$ is denoted by $\chi(t)\in C$ and we denote the glues on the north, east, west, and south edges of $t$ by $t(N)$, $t(E)$, $t(W)$, and $t(S)$, respectively.
We also call the south and west glues the {\em inputs} of $t$ while the north and east glues are called {\em outputs} of $t$; this notation will become clear in the next paragraph.
Wang tiles are not allowed to rotate.

A {\em rectilinear tile assembly system} (RTAS) $(T,\sigma)$ over $C$ and $\Sigma$ consists of a set of colored Wang tiles $T\sse C\times \Sigma^4$ and an $L$-shaped seed $\sigma$.
The seed $\sigma$ covers positions $(0,0)$ to $(m,0)$ and $(0,1)$ to $(0,n)$ of a two-dimensional Cartesian grid and it has north glues from $\Sigma$ on the positions $(1,0)$ to $(m,0)$ and east glues from $\Sigma$ on positions $(0,1)$ to $(0,n)$.
We will frequently call $T$ an RTAS without explicitly mentioning the seed, but we keep in mind that a unique seed is assigned to each RTAS.
The RTAS $T$ describes the self-assembly of a structure: starting with the seed, a tile $t$ from $T$ can attach to the structure at position $(x,y)\in[m]\times[n]$, if its west neighbor at position $(x-1,y)$ and south neighbor at position $(x,y-1)$ are present and the inputs of $t$ match the adjacent outputs of its south and west neighbors; the self-assembly stops when no more tiles in $T$ can be attached by this rule.
Arbitrarily many copies of a each tile type in $T$ are considered to be present while the structure is self-assembled, thus, one tile type can appear in multiple positions.
Fig.~\ref{fig:ex:pattern} shows the process of self-assembling a pattern by an RTAS with three tiles.
We are only interested in structures that fully tile the rectangle that is spanned by the seed.
A {\em tile assignment} in $T$ is a function $f\colon [m]\times[n] \to T$ such that $f(x,y)(W) = f(x-1,y)(E)$ and $f(x,y)(S) = f(x,y-1)(N)$ for $(x,y)\in[m]\times[n]$.
The RTAS self-assembles a pattern $P$ if there is a tile assignment $f$ in $T$ such that the color of each tile in the assignment $f$ is the color of the corresponding position in $P$, \ie $\chi \circ f = P$.
A terminological convention is to call the elements in $T$ {\em tile types} while the elements in a tile assignment are called {\em tiles}; each tile in a tile assignment is the copy a tile type from the corresponding RTAS $T$.

A {\em directed RTAS} (\tas) $T$ is an RTAS where any two distinct tile types $t_1,t_2\in T$ have different inputs, \ie $t_1(S) \neq t_2(S)$ or $t_1(W) \neq t_2(W)$.
A \tas has at most one tile assignment and can self-assemble at most one pattern.
If $T$ self-assembles an $m\times n$-pattern $P$, it defines the function $P_T\colon [m]\times[n] \to T$ such that $P_T(x,y)$ is the tile in position $(x,y)$ of the tile assignment given by $T$.
In this paper, we investigate minimal RTASs which uniquely self-assemble one given pattern $P$.
An observation from \cite{GoosO2011} allows us to focus on \tas only when searching for minimal RTAS that uniquely self-assemble a given pattern:

\begin{proposition}
If a pattern $P$ can uniquely be self-assembled by an RTAS $T$ with $m$ tile types, then there is also a \tas $T'$ with $m$ tile types which (uniquely) self-assembles $P$.
\end{proposition}

\begin{remark}\label{rem:glues}
In the following proofs, a central concept is to show that the design of a pattern $P$ enforces that any \tas $T$ which self-assembles $P$ (and maybe respects some tile type bounds) contains {\em tile types of a certain form}.
As we are flexible with choosing the set of glues, we can always obtain a \tas with different tile types by applying some bijection on the set of glues.
A subtler point is that glues used on horizontal edges and glues used on vertical edges can be seen as separate sets of glues as these edges can never glue to each other.
For the ease of notation, we will use the same glue labels for horizontal and vertical glues, but keep in mind that we may apply one bijection to all horizontal glues and another bijection to all vertical glues in order to obtain an {\em isomorphic} \tas.
\end{remark}

\section{\texorpdfstring{\NP-hardness of \pats}{NP-hardness of PATS}}
\label{sec:pats}

In this section, we prove the \NP-hardness of \pats.
The proof we present uses many techniques that have already been employed in \cite{CzeizlerP2012,shin}.
Let us also point out that we do not intend to minimize the number of colors used in our patterns or the size of our patterns.
Our motivation is to give a proof that is easy to understand and serves well as a foundation for the results in Sect.~\ref{sec:bwg}.

A boolean formula $F$ over variables $V$ in {\em conjunctive normal form with three literals per clause}, 3-CNF for short, is a boolean formula such that
\[
	F = ( c_{1,1} \lor c_{1,2} \lor c_{1,3}) \land 
		( c_{2,1} \lor c_{2,2} \lor c_{2,3}) \land \cdots \land
		( c_{\ell,1} \lor c_{\ell,2} \lor c_{\ell,3})
\]
where $c_{i,j} \in \sett{v,\neg v}{v\in V}$ for $i\in[\ell]$ and $j=1,2,3$.
It is well known that the problem \sat, to decide whether or not a given formula $F$ in 3-CNF is satisfiable, is \NP-complete; see \eg \cite{Papadimitriou2003}.
The \NP-hardness of \pats follows by the polynomial-time reduction from \sat to \pats, stated in Theorem~\ref{thm:reduction:pats}.

\begin{theorem}\label{thm:reduction:pats}
For every formula $F$ in 3-CNF there exists a pattern $P_F$ such that $F$ is satisfiable if and only if $P_F$ can be self-assembled by a \tas with at most $\abs{C(P_F)} +3$ tile types.
Moreover, $P_F$ can be computed from $F$ in polynomial time.
\end{theorem}

Theorem~\ref{thm:reduction:pats} follows by Lemmas~\ref{lem:3sat:poly} and~\ref{lem:pats}, which are presented in the following.

\begin{corollary}\label{cor:pats}
\pats is \NP-hard.
\end{corollary}

The pattern $P_F$ consists of several rectangular {\em subpatterns} which we will describe in the following.
None of the subpatterns will be adjacent to another subpattern.
The remainder of the pattern $P_F$ is filled with {\em unique colors}; a color $c$ is unique in a pattern $P$ if it appears only in one position in $P$, \ie $\abs{P^{-1}(c)} = 1$.
As a technicality that will become useful only in the proof of Theorem~\ref{thm:reduction}, we require that each position adjacent to the $L$-shaped seed or to the north or east border of pattern $P_F$ has a unique color.
Clearly, for each unique color in $P_F$ we require exactly one tile in any \tas which self-assembles $P_F$.
Since each subpattern is surrounded by a frame of unique colors, the subpatterns can be treated as if each of them would be adjacent to an $L$-shaped seed and we do not have to care about the glues on the north border or east border of a subpattern.
The placement of the tiles with unique colors is simple, as for each unique color we find a path of unique colors to the seed, using west and south steps, and we may assume that this path uses unique glues (glues which are not used anywhere else in the tile assignment).

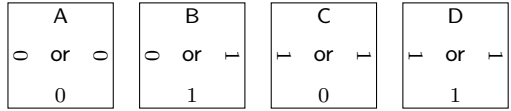
\begin{figure}[th]
	\centering
	\begin{tikzpicture}
			[scale=\fullscale,every node/.style={scale=\fullscalenode}]
		\fulltile{0}{0}{white}{black}{\tl{or}}{\tl A}{0}{0}{0}
		\fulltile{1.75}{0}{white}{black}{\tl{or}}{\tl B}{1}{1}{0}
		\fulltile{3.5}{0}{white}{black}{\tl{or}}{\tl C}{1}{0}{1}
		\fulltile{5.25}{0}{white}{black}{\tl{or}}{\tl D}{1}{1}{1}
	\end{tikzpicture}
	\caption{The four tile types used to implement the \tl{or}-gate.}
	\label{fig:or-tiles}
\end{figure}

As stated earlier, the number of tile types $m$ that is required to self-assemble $P_F$, if $F$ is satisfiable, is $m = \abs{C(P_F)} +3$.
Actually, every color in $C(P_F)$ will require one tile type only except for one color which is meant to implement an \tl{or}-gate; see Fig.~\ref{fig:or-tiles}.
Each of the tile types with color \Cor is supposed to have west input $w\in\set{0,1}$, south input $s\in \set{0,1}$, east output $w \lor s$, and an independent north output.

Our first subpattern $p$, shown in Fig.~\ref{fig:sub:p}, ensures that every \tas which self-assembles the subpattern $p$ contains at least three tile types with color \Cor.
For the upcoming proof of Theorem~\ref{thm:reduction} we need a more precise observation which draws a connection between the number of distinct output glues and the number of distinct tile types with color \Cor.

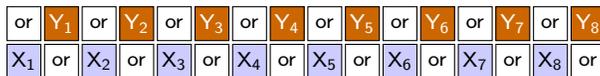
\begin{figure}[th]
	\centering
	\begin{tikzpicture}[scale=\smallscale]
		\begin{scope}
			\foreach \i in {1,...,7,8}
			{
				\smalltile{\i cm-1}{0}{colAux}{\tl{X_{\i}}}
				\smalltile{\i cm-.5}{0}{white}{\tl{or}}
				\smalltile{\i cm-1}{.5}{white}{\tl{or}}
				\smalltile{\i cm-.5}{.5}{colMod}{\color{white}\tl{Y_{\i}}}
			}
		\end{scope}
	\end{tikzpicture}
	\caption{The subpattern $p$.}
	\label{fig:sub:p}
\end{figure}

\begin{lemma}\label{lem:p_or}
A \tas $T$ which self-assembles a pattern including the subpattern $p$ contains either
\begin{compactenum}[i.)]
	\item three distinct tile types $o_1,o_2,o_3\in T$ with color \Cor all having distinct north and east glues,
	\item four distinct tile types $o_1,o_2,o_3,o_4\in T$ with color \Cor all having distinct north glues and together having at least two distinct east glues, 
	\item four distinct tile types $o_1,o_2,o_3,o_4\in T$ with color \Cor all having distinct east glues and together having at least two north glues, or
	\item eight distinct tile types $o_1,\ldots,o_8\in T$ with color \Cor all having distinct east or north glues.
\end{compactenum}
\end{lemma}

\begin{proof}
In the subpattern $p$, each of the eight tiles $t_1,\ldots,t_8$ with colors \colT{colMod}{\color{white}\tl{Y_1}} to \colT{colMod}{\color{white}\tl{Y_8}} has south and west neighbors colored by \Cor.
Since these tiles have mutually distinct colors they are all of different types  as their inputs (south and west edges) have to differ.
Therefore, the combination of outputs (north and east edges) of the tile types with color \Cor cannot be less than eight.

More formally, let $O\sse T$ be the set of tiles with color \Cor.
Let $i = \abs{\sett{t(N)}{t\in O}}$ be the number of distinct north glues on tiles from $O$ and let $j = \abs{\sett{t(E)}{t\in O}}$ be the number of distinct east glues.
If $i\cdot j$ were less than $8$, at least the inputs of two of the eight tiles $t_1,\ldots,t_8$ would coincide as their placement solely depends on the outputs of tiles from $O$.
There are four possibilities:
\begin{compactenum}[\it i.)]
	\item $i\ge 3$ and $j\ge 3$, therefore, $\abs O \ge 3$,
	\item $i\ge 4$ and $j = 2$, therefore, $\abs O \ge 4$,
	\item $i = 2$ and $j\ge 4$, therefore, $\abs O \ge 4$, or
	\item $i\ge 8$ or $j\ge 8$, therefore, $\abs O \ge 8$.
\end{compactenum}
\end{proof}

We aim to have statement {\it ii.)}\ of Lemma~\ref{lem:p_or} satisfied, but so far all four statements are possible.
Note that this lemma is independent of the number of tile types in the \tas $T$, which is a crucial difference to the observations that will follow.
The subpatterns $q_1$ to $q_5$ in Fig.~\ref{fig:sub:q} will enforce the functionality of the \tl{or}-gate tile types.

\begin{figure}[th]
	\centering
	\begin{tikzpicture}[scale=\smallscale]
		\foreach \i/\w/\s/\val/\sym/\L in {1/0/0/colNeg/-/A,2/0/1/colPos/+/B,3/1/0/colPos/+/C,4/1/1/colPos/+/D}
		{
		\begin{scope}[xshift=\i*3cm - 2.5cm]
			\smalltile{0}{0}{colAux}{\tl{Z_1}}
			\smalltile{.5}{0}{colAux}{\tl{Z_2}}
			\smalltile{0}{.5}{colAux}{\tl{Z_3}}
			\smalltile{.5}{.5}{colAux}{\tl a}

			\smalltile{0}{1}{colModd}{\color{white}$\overset{\w}\rightarrow$}
			\smalltile{.5}{1}{white}{$\overset{\w}\rightarrow$}
			\smalltile{1}{0}{colModd}{\color{white}$\uparrow\!_{\s}$}
			\smalltile{1}{.5}{white}{$\uparrow\!_{\s}$}
			\smalltile{1}{1}{white}{\tl{or}}
			
			\smalltile{1.5}{0}{colModd}{\color{white}\tl b}
			\smalltile{1.5}{.5}{colAux}{\tl b}
			\smalltile{1.5}{1}{\val}{\color{white}$\pmb{\sym}$}
			\smalltile{1.5}{1.5}{colAux}{\tl{d}}

			\smalltile{0}{1.5}{colModd}{\color{white}\tl c}
			\smalltile{.5}{1.5}{colAux}{\tl c}
			\smalltile{1}{1.5}{colMod}{\color{white}\tl{\L}}
			
			\node at (.75,-.5) {$q_{\i}$};
		\end{scope}
		}
		\begin{scope}[xshift = 4.5cm,yshift=-2.25cm]
			\smalltile{0}{0}{colAux}{\tl{a}}
			\smalltile{.5}{0}{white}{$\uparrow\!_{0}$}
			\smalltile{1}{0}{white}{$\uparrow\!_{1}$}
			\smalltile{1.5}{0}{white}{$\uparrow\!_{0}$}
			\smalltile{2}{0}{white}{$\uparrow\!_{1}$}
			\smalltile{2.5}{0}{colAux}{\tl{b}}
			
			\smalltile{0}{.5}{white}{$\overset{0}\rightarrow$}
			\smalltile{.5}{.5}{white}{\tl{or}}
			\smalltile{1}{.5}{white}{\tl{or}}
			\smalltile{1.5}{.5}{white}{\tl{or}}
			\smalltile{2}{.5}{white}{\tl{or}}
			\smalltile{2.5}{.5}{colPos}{\color{white}$\pmb+$}

			\smalltile{0}{1}{colAux}{\tl c}
			\smalltile{.5}{1}{colMod}{\color{white}\tl A}
			\smalltile{1}{1}{colMod}{\color{white}\tl B}
			\smalltile{1.5}{1}{colMod}{\color{white}\tl C}
			\smalltile{2}{1}{colMod}{\color{white}\tl D}
			\smalltile{2.5}{1}{colAux}{\tl{d}}
			\node at (1.25,-.5) {$q_5$};
		\end{scope}
	\end{tikzpicture}
	\caption{The subpatterns $q_1$ to $q_5$.}
	\label{fig:sub:q}
\end{figure}
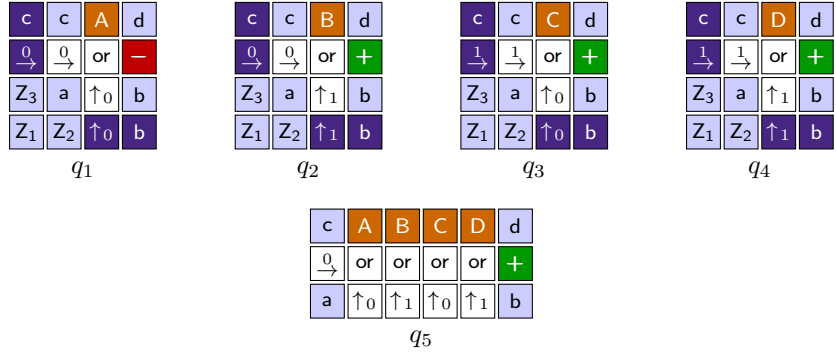

\begin{lemma}\label{lem:q_or}
Let $P$ be a pattern that contains the subpatterns $p$ and $q_1$ to $q_5$, and let $m = \abs{C(P)}+3$.
A \tas $T$ with at most $m$ tile types which self-assembles pattern $P$ contains four tile types with color \Cor of the forms shown in Fig.~\ref{fig:or-tiles}.
For every other color in $C(P)$ there exists exactly one tile type in $T$.
Moreover, the tile type with color \Cra has east output $0$
and the tile type with color \Cpos has west input $1$.
\end{lemma}

\begin{proof}
By Lemma~\ref{lem:p_or}, the \tas $T$ contains at least three tile types with color \Cor.
Since we need at least one tile type for each color in $C(P)$, there is one tile type left in $T$ whose color is not determined yet.
By $o_i\in T$ we denote the tile of color \Cor in subpattern $q_i$ for $i=1,2,3,4$.
We will show that all four tiles $o_1,o_2,o_3,o_4$ are of different types and are the tile types shown in Fig.~\ref{fig:or-tiles}, from left to right.

Suppose by contradiction, there were only three tile types with color \Cor and only one tile type with color \colT{colAux}{\tl c}.
Two of the tiles with colors \colT{colMod}{\color{white}\tl A}, \colT{colMod}{\color{white}\tl B}, \colT{colMod}{\color{white}\tl C}, and \colT{colMod}{\color{white}\tl D} in patterns $q_1$ to $q_4$ would have the same south and west neighbors and, hence, the inputs of the two corresponding tile types would coincide --- a property which is forbidden for \tas{}s.
Thus, either there are four tile types with color \Cor or two tile types with color \colT{colAux}{\tl c}.
For every other color there is exactly one tile type.

As there is only one tile type with color \colT{colAux}{\tl b} the tile $o_1$ and $o_2$ have to be of different types; otherwise, the tiles with colors \Cneg and \Cpos would have the same inputs in subpatterns $q_1$ and $q_2$.
Because tiles $o_1$ and $o_2$ have the same west neighbor, the tiles with colors \colT{white}{$\uparrow\!_0$} and \colT{white}{$\uparrow\!_1$} are responsible for the placement of $o_1$ and $o_2$, respectively;
this means the north glues of the tile types with colors \colT{white}{$\uparrow\!_0$} and \colT{white}{$\uparrow\!_1$} differ.
Symmetrically, $o_1$ and $o_3$ are of different types and the east glues of the tile types with colors \Cra and \colT{white}{$\overset1\rightarrow$} differ.
Next, we see that the four tiles $o_1$ to $o_4$ all have different inputs, thus, the types of $o_1$ to $o_4$ are mutually distinct and their is only one tile type with color \colT{colAux}{\tl c}.

By the freedom of naming the used glues, see Remark~\ref{rem:glues}, we assume that the tile type with color \Cra has east output $0$, and the tile type with \colT{white}{$\overset1\rightarrow$} has east output $1$, the tile type with color \colT{white}{$\uparrow\!_0$} has north output $0$, and the tile type with color \colT{white}{$\uparrow\!_1$} has north output $1$.
The two tiles with colors \Cpos and \Cneg have the same south input whence their west input differs and depends on the neighboring \tl{or}-gate tile.
We assume that the tile with color \Cneg has west input $0'$ and the tile with color \Cpos has west input $1'$.
We have $o_1(E) = 0'$ and $o_2(E) = o_3(E) = o_4(E) = 1'$.
By Lemma~\ref{lem:p_or}, the north outputs of the four tiles $o_1$ to $o_4$ have to be distinct.

Next, we take a look at subpattern $q_5$.
The tiles with colors \colT{colMod}{\color{white}\tl A}, \colT{colMod}{\color{white}\tl B}, \colT{colMod}{\color{white}\tl C}, and \colT{colMod}{\color{white}\tl D}, in the top row of subpattern $q_5$ enforce that the four \tl{or}-gates below are of the same types as $o_1$, $o_2$, $o_3$, and $o_4$, from left to right; otherwise the south inputs of the top row cannot match the north outputs of the middle row.
By the placement of the four \tl{or}-gate tiles it is clear that $0 = 0'$ and $1 = 1'$ as desired.
\end{proof}

\begin{figure}[th]
	\centering
	\begin{tikzpicture}[scale=\smallscale]
		\begin{scope}
			\smalltile{0}{0}{colAux}{\tl{Z_4}}
			\smalltile{0}{.5}{colAux}{$v$}
			\smalltile{.5}{0}{white}{$v$}
			\smalltile{.5}{.5}{colMod}{\color{white}$v$}

			\node at (.25,-.5) {$r_1(v)$};
		\end{scope}
		\begin{scope}[xshift=2cm]
			\smalltile{0}{0}{colAux}{\tl{Z_4}}
			\smalltile{0}{.5}{colAux}{$v$}
			\smalltile{.5}{0}{white}{$\neg v$}
			\smalltile{.5}{.5}{colModd}{\color{white}$\tilde v$}

			\node at (.25,-.5) {$r_2(v)$};
		\end{scope}

		\begin{scope}[xshift = 4cm]
			\smalltile{0}{0}{colAux}{\tl{a}}
			\smalltile{.5}{0}{white}{$v$}
			\smalltile{1}{0}{white}{$\neg v$}
			\smalltile{1.5}{0}{colAux}{\tl{b}}
			
			\smalltile{0}{.5}{white}{$\overset{0}\rightarrow$}
			\smalltile{.5}{.5}{white}{\tl{or}}
			\smalltile{1}{.5}{white}{\tl{or}}
			\smalltile{1.5}{.5}{colPos}{\color{white}$\pmb+$}

			\node at (.75,-.5) {$r_3(v)$};
		\end{scope}
		\begin{scope}[xshift=7cm]
			\smalltile{0}{0}{colAux}{\tl{a}}
			\smalltile{.5}{0}{white}{$c_1$}
			\smalltile{1}{0}{white}{$c_2$}
			\smalltile{1.5}{0}{white}{$c_3$}
			\smalltile{2}{0}{colAux}{\tl{b}}
			
			\smalltile{0}{.5}{white}{$\overset{0}\rightarrow$}
			\smalltile{.5}{.5}{white}{\tl{or}}
			\smalltile{1}{.5}{white}{\tl{or}}
			\smalltile{1.5}{.5}{white}{\tl{or}}
			\smalltile{2}{.5}{colPos}{\color{white}$\pmb+$}

			\node at (1,-.5) {$s(C)$};
		\end{scope}
	\end{tikzpicture}
	\caption{The subpatterns $r_1(v)$ to $r_3(v)$ for a variable $v\in V$ and
	the subpattern $s(C)$ for a clause 
	$C = (c_1\lor c_2 \lor c_3)$ in $F$ where $c_i = v$ or $c_i = \neg v$
	for some variable $v\in V$ and $i=1,2,3$.}
	\label{fig:sub:F}
\end{figure}
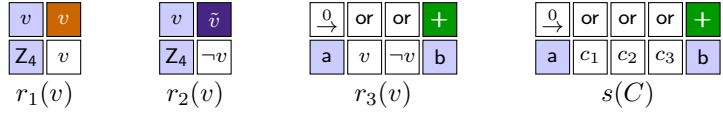

The subpatterns that we defined so far did not depend on the formula $F$.
Now, for each variable $v\in V$ we define three subpatterns $r_1(v)$, $r_2(v)$, $r_3(v)$ and for a clause $C$ from $F$ we define one more subpattern $s(C)$; these patterns are given by Fig.~\ref{fig:sub:F}.
For a formula $F$ in 3-CNF we let $P_F$ be the pattern that contains all the subpatterns $p$, $q_1$ to $q_5$, $r_1(v)$ to $r_3(v)$ for each variable $v\in V$, and $s(C)$ for each clause $C$ from $F$, where each subpattern is placed to the right of the previous subpattern with one column of unique colors in between.
Then, $P_F$ has height $6$, because the top and bottom rows contain unique colors only, and $P_F$ has width $45+11\cdot \abs V + 6\cdot \ell$.
The next lemma follows from this observation.

\begin{lemma}\label{lem:3sat:poly}
Given a formula $F$ in 3-CNF, the pattern $P_F$ can be computed from $F$ in polynomial time.
\end{lemma}

\begin{proof}
This is obvious by the design of the pattern.
\end{proof}

The subpatterns $r_1(v)$ and $r_2(v)$ ensure that the two tile types with colors \colT{white}{$v$} and \colT{white}{$\neg v$} have distinct north outputs.
The subpattern $r_3(v)$ then implies that one of the north glues is $0$ and the other one is $1$.

\begin{lemma}\label{lem:sub:r}
Let $P_F$ be the pattern for a formula $F$ over variables $V$ in 3-CNF and let $T$ be a \tas with at most $m = \abs{C(P_F)}+3$ tile types which self-assembles pattern $P_F$.
For all variables $v\in V$, there is a unique tile type $t_v^\oplus\in T$ with color \colT{white}{$v$} and a unique tile type $t_v^\ominus\in T$ with color \colT{white}{$\neg v$} such that either $t_v^\oplus(N) = 1$ and $t_v^\ominus(N) = 0$ or $t_v^\oplus(N) = 0$ and $t_v^\ominus(N) = 1$.
Here, $0$ and $1$ are the south inputs of the \tl{or}-gate tile types as shown in Fig.~\ref{fig:or-tiles}.
\end{lemma}

\begin{proof}
Let $v\in V$, $t_v^\oplus\in T$ be the tile type with color \colT{white}{$v$}, and $t_v^\ominus\in T$ be the tile type with color \colT{white}{$\neg v$}.
The fact that $t_v^\oplus$ and $t_v^\ominus$ are unique is stated Lemma~\ref{lem:q_or} and, furthermore, for every color in $C(P_F)$ except for \Cor there exists just one tile type in $T$ with that color.
In particular, there is only one tile type with color \colT{colAux}{$v$}.
Since the north neighbors of the tile of type $t_v^\oplus$ in subpattern $r_1(v)$ and the tile of type $t_v^\ominus$ in subpattern $r_2(v)$ differ while their north neighbors have the same west input, we conclude that $t_v^\oplus(N) \neq t_v^\ominus(N)$.
Moreover, both of the tile types are a south neighbor of an \tl{or}-gate tile in subpattern $r_3(v)$, hence, $t_v^\oplus(N), t_v^\ominus(N)\in\set{0,1}$.
\end{proof}

Now, these glues serve as input for the \tl{or}-gates in the subpatterns $s(C)$.
The following lemma concludes the proof of Theorem~\ref{thm:reduction:pats}.

\begin{lemma}\label{lem:pats}
Let $P_F$ be the pattern for a formula $F$ over variables $V$ in 3-CNF and let $m = \abs{C(P_F)}+3$.
The formula $F$ is satisfiable if and only if $P_F$ can be self-assembled by a \tas $T$ with at most $m$ tile types.
\end{lemma}

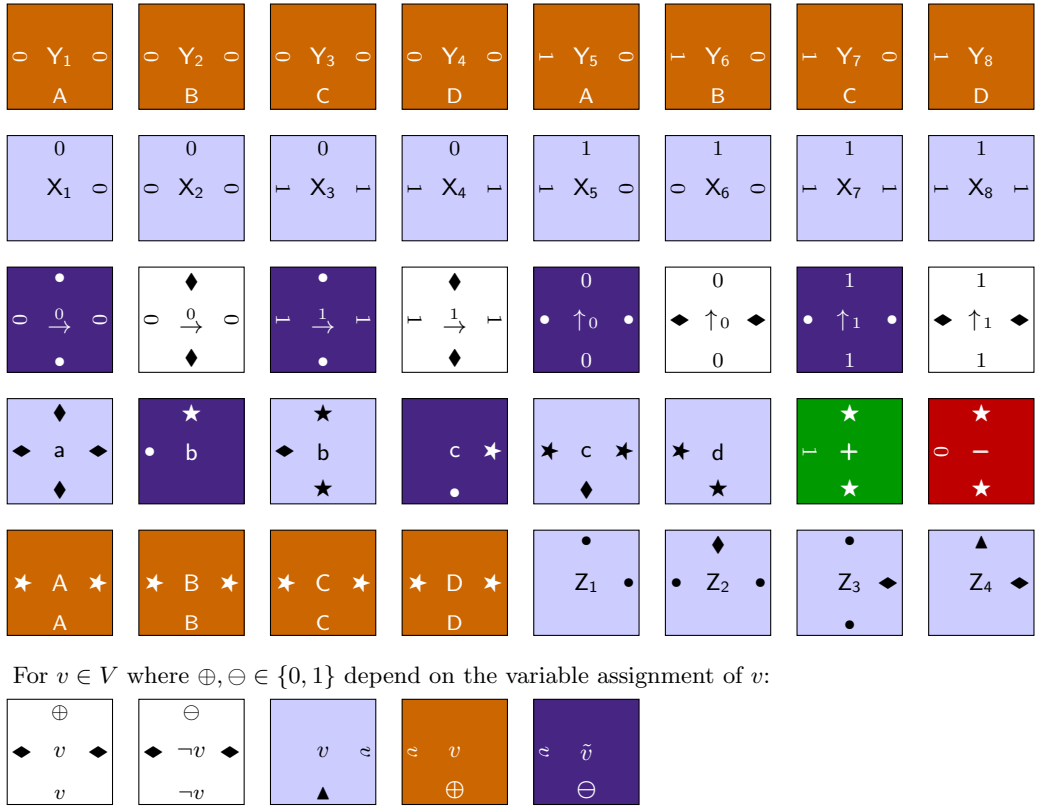
\begin{figure}[th]
	\centering
	\hspace*{-2cm}
	\begin{tikzpicture}
		[scale=\fullscale,every node/.style={scale=\fullscalenode}]
	
	
	\begin{scope}[yshift=1.75cm]
		\fulltile{0}{0}{colMod}{white}{\tl{Y_1}}{}{0}{\tl A}{0}
		\fulltile{1.75}{0}{colMod}{white}{\tl{Y_2}}{}{0}{\tl B}{0}
		\fulltile{3.5}{0}{colMod}{white}{\tl{Y_3}}{}{0}{\tl C}{0}
		\fulltile{5.25}{0}{colMod}{white}{\tl{Y_4}}{}{0}{\tl D}{0}
		\fulltile{7}{0}{colMod}{white}{\tl{Y_5}}{}{0}{\tl A}{1}
		\fulltile{8.75}{0}{colMod}{white}{\tl{Y_6}}{}{0}{\tl B}{1}
		\fulltile{10.5}{0}{colMod}{white}{\tl{Y_7}}{}{0}{\tl C}{1}
		\fulltile{12.25}{0}{colMod}{white}{\tl{Y_8}}{}{}{\tl D}{1}
	\end{scope}
	
		\fulltile{0}{0}{colAux}{black}{\tl{X_1}}{0}{0}{}{}
		\fulltile{1.75}{0}{colAux}{black}{\tl{X_2}}{0}{0}{}{0}
		\fulltile{3.5}{0}{colAux}{black}{\tl{X_3}}{0}{1}{}{1}
		\fulltile{5.25}{0}{colAux}{black}{\tl{X_4}}{0}{1}{}{1}
		\fulltile{7}{0}{colAux}{black}{\tl{X_5}}{1}{0}{}{1}
		\fulltile{8.75}{0}{colAux}{black}{\tl{X_6}}{1}{0}{}{0}
		\fulltile{10.5}{0}{colAux}{black}{\tl{X_7}}{1}{1}{}{1}
		\fulltile{12.25}{0}{colAux}{black}{\tl{X_8}}{1}{1}{}{1}

	\begin{scope}[yshift=-1.75cm]
		\fulltile{0}{0}{colModd}{white}{\overset0\rightarrow}{\bul}{0}{\bul}{0}
		\fulltile{1.75}{0}{white}{black}{\overset0\rightarrow}{\dia}{0}{\dia}{0}
		\fulltile{3.5}{0}{colModd}{white}{\overset1\rightarrow}{\bul}{1}{\bul}{1}
		\fulltile{5.25}{0}{white}{black}{\overset1\rightarrow}{\dia}{1}{\dia}{1}
		\fulltile{7}{0}{colModd}{white}{\uparrow\!_0}{0}{\bul}{0}{\bul}
		\fulltile{8.75}{0}{white}{black}{\uparrow\!_0}{0}{\dia}{0}{\dia}
		\fulltile{10.5}{0}{colModd}{white}{\uparrow\!_1}{1}{\bul}{1}{\bul}
		\fulltile{12.25}{0}{white}{black}{\uparrow\!_1}{1}{\dia}{1}{\dia}
	\end{scope}

	\begin{scope}[yshift=-3.5cm]
		\fulltile{0}{0}{colAux}{black}{\tl a}{\dia}{\dia}{\dia}{\dia}
		\fulltile{1.75}{0}{colModd}{white}{\tl b}{\str}{}{}{\bul}
		\fulltile{3.5}{0}{colAux}{black}{\tl b}{\str}{}{\str}{\dia}
		\fulltile{5.25}{0}{colModd}{white}{\tl c}{}{\str}{\bul}{}
		\fulltile{7}{0}{colAux}{black}{\tl c}{}{\str}{\dia}{\str}
		\fulltile{8.75}{0}{colAux}{black}{\tl{d}}{}{}{\str}{\str}
		\fulltile{10.5}{0}{colPos}{white}{\pmb+}{\str}{}{\str}{1}
		\fulltile{12.25}{0}{colNeg}{white}{\pmb-}{\str}{}{\str}{0}
	\end{scope}

	\begin{scope}[yshift=-5.25cm]
		\fulltile{0}{0}{colMod}{white}{\tl A}{}{\str}{\tl A}{\str}
		\fulltile{1.75}{0}{colMod}{white}{\tl B}{}{\str}{\tl B}{\str}
		\fulltile{3.5}{0}{colMod}{white}{\tl C}{}{\str}{\tl C}{\str}
		\fulltile{5.25}{0}{colMod}{white}{\tl D}{}{\str}{\tl D}{\str}
		\fulltile{7}{0}{colAux}{black}{\tl{Z_1}}{\bul}{\bul}{}{}
		\fulltile{8.75}{0}{colAux}{black}{\tl{Z_2}}{\dia}{\bul}{}{\bul}
		\fulltile{10.5}{0}{colAux}{black}{\tl{Z_3}}{\bul}{\dia}{\bul}{}
		\fulltile{12.25}{0}{colAux}{black}{\tl{Z_4}}{\tri}{\dia}{}{}
	\end{scope}

	\begin{scope}[yshift=-7.5cm]
		\node [anchor=west,font=\small] at (-.75,1) {For $v \in V$ where $\oplus,\ominus\in\set{0,1}$ depend on the variable assignment of $v$:};
		\fulltile{0}{0}{white}{black}{v}{\oplus}{\dia}{v}{\dia}
		\fulltile{1.75}{0}{white}{black}{\neg v}{\ominus}{\dia}{\neg v}{\dia}
		\fulltile{3.5}{0}{colAux}{black}{v}{}{v}{\tri}{}
		\fulltile{5.25}{0}{colMod}{white}{v}{}{}{\boldsymbol\oplus}{v}
		\fulltile{7}{0}{colModd}{white}{\tilde v}{}{}{\boldsymbol\ominus}{v}
	\end{scope}
	\end{tikzpicture}
	\hspace*{-2cm}
	\caption{All tile types, except for the \tl{or}-gate, used to self-assemble
		the subpatterns in $P_F$.
		Unlabeled edges only appear on the outside borders of subpatterns;
		unique glues on these edges can be used to attach the subpatterns
		to the remaining structure.}
	\label{fig:all:3sat:tiles}
\end{figure}

\begin{proof}
The formula $F$ is satisfiable if and only if there is a variable assignment $f\colon V\to \set{0,1}$ which satisfies every clause in $F$.
Suppose such a variable assignment $f$ exists.
For every variable $v\in V$ we let $t_v^\oplus\in T$ be the tile type with color \colT{white}{$v$} and $t_v^\ominus\in T$ be the tile type with color \colT{white}{$\neg v$}.
In accordance with Lemma~\ref{lem:sub:r}, we let $t_v^\oplus(N) = f(v)$ and $t_v^\ominus(N) = 1-f(v)$.
The remaining glues on $t_v^\oplus$ and $t_v^\ominus$ and the other tile types are given in Fig.~\ref{fig:all:3sat:tiles} plus the four \tl{or}-gate tile types in Fig.~\ref{fig:or-tiles}.
By design, it is clear that these tile types can self-assemble all the subpatterns $p$, $q_1$ to $q_5$, and $r_1(v)$ to $r_3(v)$ for $v\in V$; see also Sect.~\ref{app:subpatterns} where the subpatterns with their tile assignments are presented.
Now, consider the subpattern $s(C)$ for a clause $C=(c_1\lor c_2 \lor c_3)$ in $F$.
It is not difficult to observe that, since the clause $C$ is satisfied by the assignment $f$, at least one of the north glues of the tiles in $s(C)$ with colors \colT{white}{$c_1$}, \colT{white}{$c_2$}, and \colT{white}{$c_3$} is $1$.
The design of the \tl{or}-gate tile types ensures that the right \tl{or}-gate tile in $s(C)$ has east output $1$ and the tile with color \Cpos can be placed --- concluding that $P_F$ can be self-assembled by the given \tas having $m$ tile types.

Conversely, suppose $P_F$ can be self-assembled by a \tas with $m$ tile types.
By Lemma~\ref{lem:sub:r}, for each variable $v\in V$ the tile type $t_v^\oplus$ with color \colT{white}{$v$} has either north output $0$ or $1$.
We define a variable assignment $f\colon V \to \set{0,1}$ by $v\mapsto t_v^\oplus(N)$.
Recall from Lemma~\ref{lem:q_or} that the four \tl{or}-gate tile types in $T$ actually implement an \tl{or}-gate.
Furthermore, in a subpattern $s(C)$ where $C = (c_1\lor c_2\lor c_3)$ is a clause from $F$, the west input of the left \tl{or}-gate tile is $0$, and if the east output of the right \tl{or}-gate tile were $0$, then the tile type with color \Cpos could not be placed in $s(C)$ (instead the tile type with color \Cneg would be placed).
Hence, at least one of the north glues of the tiles with colors \colT{white}{$c_1$}, \colT{white}{$c_2$}, and \colT{white}{$c_3$} is $1$.
Using the fact that for $v\in V$ the north output of the tile type with color \colT{white}{$\neg v$} is $1$ if and only if the north output of the tile type with color \colT{white}{$v$} is $0$, stated in Lemma~\ref{lem:sub:r}, we infer that $f$ is a variable assignment that satisfies every clause in $F$.
Therefore, $F$ is satisfiable.
\end{proof}

\subsection{Subpatterns in \texorpdfstring{$P_F$}{P(F)} with Tile Assignment}
\label{app:subpatterns}

In this section we present the subpatterns $p$, $q_1$ to $q_5$, $r_1(v)$ to $r_3(v)$, $r_1(u)$ to $r_3(u)$, and $s(C)$ in $P_F$ with tile assignments, where $C = ( u \lor v \lor \neg w)$ and $u,v,w\in V$ with the variable assignment $u\mapsto 0$, $v\mapsto 1$, and $w\mapsto 1$.
The used tiles are described in Fig.~\ref{fig:or-tiles} and~\ref{fig:all:3sat:tiles}.
The other subpatterns in $P_F$ can be tiled analogously.

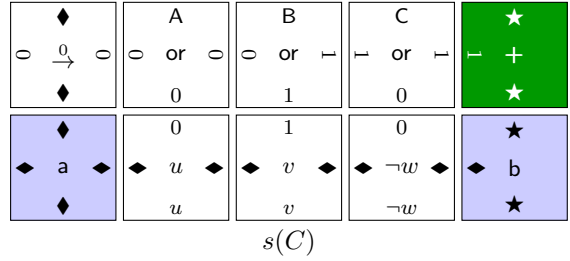
\begin{center}
	\hspace*{-2cm}
	\begin{tikzpicture}
			[scale=\fullscale,every node/.style={scale=\fullscalenode}]
		\begin{scope}[yshift=1.5cm]	
		\fulltile{0}{0}{white}{black}{\tl{or}}{\tl A}{0}{0}{0}
		\fulltile{1.5}{0}{colMod}{white}{\tl{Y_1}}{}{0}{\tl A}{0}
		\fulltile{3}{0}{white}{black}{\tl{or}}{\tl A}{0}{0}{0}
		\fulltile{4.5}{0}{colMod}{white}{\tl{Y_2}}{}{0}{\tl B}{0}
		\fulltile{6}{0}{white}{black}{\tl{or}}{\tl A}{0}{0}{0}
		\fulltile{7.5}{0}{colMod}{white}{\tl{Y_3}}{}{0}{\tl C}{0}
		\fulltile{9}{0}{white}{black}{\tl{or}}{\tl A}{0}{0}{0}
		\fulltile{10.5}{0}{colMod}{white}{\tl{Y_4}}{}{0}{\tl D}{0}
		\end{scope}

		\fulltile{0}{0}{colAux}{black}{\tl{X_1}}{0}{0}{}{}
		\fulltile{1.5}{0}{white}{black}{\tl{or}}{\tl A}{0}{0}{0}
		\fulltile{3}{0}{colAux}{black}{\tl{X_2}}{0}{0}{}{0}
		\fulltile{4.5}{0}{white}{black}{\tl{or}}{\tl B}{1}{1}{0}
		\fulltile{6}{0}{colAux}{black}{\tl{X_3}}{0}{1}{}{1}
		\fulltile{7.5}{0}{white}{black}{\tl{or}}{\tl C}{1}{0}{1}
		\fulltile{9}{0}{colAux}{black}{\tl{X_4}}{0}{1}{}{1}
		\fulltile{10.5}{0}{white}{black}{\tl{or}}{\tl D}{1}{1}{1}

		\begin{scope}[yshift=-3.5cm, xshift=-10.5cm]
		\begin{scope}[yshift=1.5cm]	
		\fulltile{12}{0}{white}{black}{\tl{or}}{\tl B}{1}{1}{0}
		\fulltile{13.5}{0}{colMod}{white}{\tl{Y_5}}{}{0}{\tl A}{1}
		\fulltile{15}{0}{white}{black}{\tl{or}}{\tl B}{1}{1}{0}
		\fulltile{16.5}{0}{colMod}{white}{\tl{Y_6}}{}{0}{\tl B}{1}
		\fulltile{18}{0}{white}{black}{\tl{or}}{\tl B}{1}{1}{0}
		\fulltile{19.5}{0}{colMod}{white}{\tl{Y_7}}{}{0}{\tl C}{1}
		\fulltile{21}{0}{white}{black}{\tl{or}}{\tl B}{1}{1}{0}
		\fulltile{22.5}{0}{colMod}{white}{\tl{Y_8}}{}{}{\tl D}{1}
		\end{scope}

		\fulltile{12}{0}{colAux}{black}{\tl{X_5}}{1}{0}{}{1}
		\fulltile{13.5}{0}{white}{black}{\tl{or}}{\tl A}{0}{0}{0}
		\fulltile{15}{0}{colAux}{black}{\tl{X_6}}{1}{0}{}{0}
		\fulltile{16.5}{0}{white}{black}{\tl{or}}{\tl B}{1}{1}{0}
		\fulltile{18}{0}{colAux}{black}{\tl{X_7}}{1}{1}{}{1}
		\fulltile{19.5}{0}{white}{black}{\tl{or}}{\tl C}{1}{0}{1}
		\fulltile{21}{0}{colAux}{black}{\tl{X_8}}{1}{1}{}{1}
		\fulltile{22.5}{0}{white}{black}{\tl{or}}{\tl D}{1}{1}{1}
		\end{scope}
		
		\node at (12,.75) {$\ldots$};
		\node at (0,-2.75) {$\ldots$};
		\node at (6,-4.5) {$p$};
	\end{tikzpicture}
	\hspace*{-2cm}
	
	\bigskip
	
	\hspace*{-2cm}
	\begin{tikzpicture}
			[scale=\fullscale,every node/.style={scale=\fullscalenode}]
		\fulltile{0}{0}{colAux}{black}{\tl{Z_1}}{\bul}{\bul}{}{}
		\fulltile{1.5}{0}{colAux}{black}{\tl{Z_2}}{\dia}{\bul}{}{\bul}
		\fulltile{3}{0}{colModd}{white}{\uparrow\!_0}{0}{\bul}{0}{\bul}
		\fulltile{4.5}{0}{colModd}{white}{\tl b}{\str}{}{}{\bul}

		\fulltile{0}{1.5}{colAux}{black}{\tl{Z_3}}{\bul}{\dia}{\bul}{}
		\fulltile{1.5}{1.5}{colAux}{black}{\tl a}{\dia}{\dia}{\dia}{\dia}
		\fulltile{3}{1.5}{white}{black}{\uparrow\!_0}{0}{\dia}{0}{\dia}
		\fulltile{4.5}{1.5}{colAux}{black}{\tl b}{\str}{}{\str}{\dia}

		\fulltile{0}{3}{colModd}{white}{\overset0\rightarrow}{\bul}{0}{\bul}{0}
		\fulltile{1.5}{3}{white}{black}{\overset0\rightarrow}{\dia}{0}{\dia}{0}
		\fulltile{3}{3}{white}{black}{\tl{or}}{\tl A}{0}{0}{0}
		\fulltile{4.5}{3}{colNeg}{white}{\pmb-}{\str}{}{\str}{0}

		\fulltile{0}{4.5}{colModd}{white}{\tl c}{}{\str}{\bul}{}
		\fulltile{1.5}{4.5}{colAux}{black}{\tl c}{}{\str}{\dia}{\str}
		\fulltile{3}{4.5}{colMod}{white}{\tl A}{}{\str}{\tl A}{\str}
		\fulltile{4.5}{4.5}{colAux}{black}{\tl{d}}{}{}{\str}{\str}

		\node at (2.25,-1) {$q_1$};
	
	\begin{scope}[xshift=7cm]
			[scale=\fullscale,every node/.style={scale=\fullscalenode}]
		\fulltile{0}{0}{colAux}{black}{\tl{Z_1}}{\bul}{\bul}{}{}
		\fulltile{1.5}{0}{colAux}{black}{\tl{Z_2}}{\dia}{\bul}{}{\bul}
		\fulltile{3}{0}{colModd}{white}{\uparrow\!_1}{1}{\bul}{1}{\bul}
		\fulltile{4.5}{0}{colModd}{white}{\tl b}{\str}{}{}{\bul}

		\fulltile{0}{1.5}{colAux}{black}{\tl{Z_3}}{\bul}{\dia}{\bul}{}
		\fulltile{1.5}{1.5}{colAux}{black}{\tl a}{\dia}{\dia}{\dia}{\dia}
		\fulltile{3}{1.5}{white}{black}{\uparrow\!_1}{1}{\dia}{1}{\dia}
		\fulltile{4.5}{1.5}{colAux}{black}{\tl b}{\str}{}{\str}{\dia}

		\fulltile{0}{3}{colModd}{white}{\overset0\rightarrow}{\bul}{0}{\bul}{0}
		\fulltile{1.5}{3}{white}{black}{\overset0\rightarrow}{\dia}{0}{\dia}{0}
		\fulltile{3}{3}{white}{black}{\tl{or}}{\tl B}{1}{1}{0}
		\fulltile{4.5}{3}{colPos}{white}{\pmb+}{\str}{}{\str}{1}

		\fulltile{0}{4.5}{colModd}{white}{\tl c}{}{\str}{\bul}{}
		\fulltile{1.5}{4.5}{colAux}{black}{\tl c}{}{\str}{\dia}{\str}
		\fulltile{3}{4.5}{colMod}{white}{\tl B}{}{\str}{\tl B}{\str}
		\fulltile{4.5}{4.5}{colAux}{black}{\tl{d}}{}{}{\str}{\str}

		\node at (2.25,-1) {$q_2$};
	\end{scope}
	\end{tikzpicture}
	\hspace*{-2cm}

	\bigskip
	
	\hspace*{-2cm}
	\begin{tikzpicture}
			[scale=\fullscale,every node/.style={scale=\fullscalenode}]
		\fulltile{0}{0}{colAux}{black}{\tl{Z_1}}{\bul}{\bul}{}{}
		\fulltile{1.5}{0}{colAux}{black}{\tl{Z_2}}{\dia}{\bul}{}{\bul}
		\fulltile{3}{0}{colModd}{white}{\uparrow\!_0}{0}{\bul}{0}{\bul}
		\fulltile{4.5}{0}{colModd}{white}{\tl b}{\str}{}{}{\bul}

		\fulltile{0}{1.5}{colAux}{black}{\tl{Z_3}}{\bul}{\dia}{\bul}{}
		\fulltile{1.5}{1.5}{colAux}{black}{\tl a}{\dia}{\dia}{\dia}{\dia}
		\fulltile{3}{1.5}{white}{black}{\uparrow\!_0}{0}{\dia}{0}{\dia}
		\fulltile{4.5}{1.5}{colAux}{black}{\tl b}{\str}{}{\str}{\dia}

		\fulltile{0}{3}{colModd}{white}{\overset1\rightarrow}{\bul}{1}{\bul}{1}
		\fulltile{1.5}{3}{white}{black}{\overset1\rightarrow}{\dia}{1}{\dia}{1}
		\fulltile{3}{3}{white}{black}{\tl{or}}{\tl C}{1}{0}{1}
		\fulltile{4.5}{3}{colPos}{white}{\pmb-}{\str}{}{\str}{1}

		\fulltile{0}{4.5}{colModd}{white}{\tl c}{}{\str}{\bul}{}
		\fulltile{1.5}{4.5}{colAux}{black}{\tl c}{}{\str}{\dia}{\str}
		\fulltile{3}{4.5}{colMod}{white}{\tl C}{}{\str}{\tl C}{\str}
		\fulltile{4.5}{4.5}{colAux}{black}{\tl{d}}{}{}{\str}{\str}

		\node at (2.25,-1) {$q_1$};
	
	\begin{scope}[xshift=7cm]
			[scale=\fullscale,every node/.style={scale=\fullscalenode}]
		\fulltile{0}{0}{colAux}{black}{\tl{Z_1}}{\bul}{\bul}{}{}
		\fulltile{1.5}{0}{colAux}{black}{\tl{Z_2}}{\dia}{\bul}{}{\bul}
		\fulltile{3}{0}{colModd}{white}{\uparrow\!_1}{1}{\bul}{1}{\bul}
		\fulltile{4.5}{0}{colModd}{white}{\tl b}{\str}{}{}{\bul}

		\fulltile{0}{1.5}{colAux}{black}{\tl{Z_3}}{\bul}{\dia}{\bul}{}
		\fulltile{1.5}{1.5}{colAux}{black}{\tl a}{\dia}{\dia}{\dia}{\dia}
		\fulltile{3}{1.5}{white}{black}{\uparrow\!_1}{1}{\dia}{1}{\dia}
		\fulltile{4.5}{1.5}{colAux}{black}{\tl b}{\str}{}{\str}{\dia}

		\fulltile{0}{3}{colModd}{white}{\overset1\rightarrow}{\bul}{1}{\bul}{1}
		\fulltile{1.5}{3}{white}{black}{\overset1\rightarrow}{\dia}{1}{\dia}{1}
		\fulltile{3}{3}{white}{black}{\tl{or}}{\tl D}{1}{1}{1}
		\fulltile{4.5}{3}{colPos}{white}{\pmb+}{\str}{}{\str}{1}

		\fulltile{0}{4.5}{colModd}{white}{\tl c}{}{\str}{\bul}{}
		\fulltile{1.5}{4.5}{colAux}{black}{\tl c}{}{\str}{\dia}{\str}
		\fulltile{3}{4.5}{colMod}{white}{\tl D}{}{\str}{\tl D}{\str}
		\fulltile{4.5}{4.5}{colAux}{black}{\tl{d}}{}{}{\str}{\str}

		\node at (2.25,-1) {$q_2$};
	\end{scope}
	\end{tikzpicture}
	\hspace*{-2cm}

	\bigskip
	
	\begin{tikzpicture}
			[scale=\fullscale,every node/.style={scale=\fullscalenode}]
		\fulltile{0}{0}{colAux}{black}{\tl a}{\dia}{\dia}{\dia}{\dia}
		\fulltile{1.5}{0}{white}{black}{\uparrow\!_0}{0}{\dia}{0}{\dia}
		\fulltile{3}{0}{white}{black}{\uparrow\!_1}{1}{\dia}{1}{\dia}
		\fulltile{4.5}{0}{white}{black}{\uparrow\!_0}{0}{\dia}{0}{\dia}
		\fulltile{6}{0}{white}{black}{\uparrow\!_1}{1}{\dia}{1}{\dia}
		\fulltile{7.5}{0}{colAux}{black}{\tl b}{\str}{}{\str}{\dia}
		
		\fulltile{0}{1.5}{white}{black}{\overset0\rightarrow}{\dia}{0}{\dia}{0}
		\fulltile{1.5}{1.5}{white}{black}{\tl{or}}{\tl A}{0}{0}{0}
		\fulltile{3}{1.5}{white}{black}{\tl{or}}{\tl B}{1}{1}{0}
		\fulltile{4.5}{1.5}{white}{black}{\tl{or}}{\tl C}{1}{0}{1}
		\fulltile{6}{1.5}{white}{black}{\tl{or}}{\tl D}{1}{1}{1}
		\fulltile{7.5}{1.5}{colPos}{white}{\pmb+}{\str}{}{\str}{1}

		\fulltile{0}{3}{colAux}{black}{\tl c}{}{\str}{\dia}{\str}
		\fulltile{1.5}{3}{colMod}{white}{\tl A}{}{\str}{\tl A}{\str}
		\fulltile{3}{3}{colMod}{white}{\tl B}{}{\str}{\tl B}{\str}
		\fulltile{4.5}{3}{colMod}{white}{\tl C}{}{\str}{\tl C}{\str}
		\fulltile{6}{3}{colMod}{white}{\tl D}{}{\str}{\tl D}{\str}
		\fulltile{7.5}{3}{colAux}{black}{\tl{d}}{}{}{\str}{\str}

		\node at (3.75,-1) {$q_5$};
	\end{tikzpicture}

	\bigskip

	\hspace*{-2cm}
	\begin{tikzpicture}
			[scale=\fullscale,every node/.style={scale=\fullscalenode}]

	\begin{scope}
		\fulltile{0}{0}{colAux}{black}{\tl{Z_4}}{\tri}{\dia}{}{}
		\fulltile{1.5}{0}{white}{black}{v}{1}{\dia}{v}{\dia}

		\fulltile{0}{1.5}{colAux}{black}{v}{}{v}{\tri}{}
		\fulltile{1.5}{1.5}{colMod}{white}{v}{}{}{1}{v}
		
		\node at (.75,-1) {$r_1(v)$};
	\end{scope}

	\begin{scope}[xshift=4cm]
		\fulltile{0}{0}{colAux}{black}{\tl{Z_4}}{\tri}{\dia}{}{}
		\fulltile{1.5}{0}{white}{black}{\neg v}{0}{\dia}{\neg v}{\dia}

		\fulltile{0}{1.5}{colAux}{black}{v}{}{v}{\tri}{}
		\fulltile{1.5}{1.5}{colModd}{white}{\tilde v}{}{}{0}{v}
		
		\node at (.75,-1) {$r_2(v)$};
	\end{scope}
	
	\begin{scope}[xshift=8cm]
		\fulltile{0}{0}{colAux}{black}{\tl a}{\dia}{\dia}{\dia}{\dia}
		\fulltile{1.5}{0}{white}{black}{v}{1}{\dia}{v}{\dia}
		\fulltile{3}{0}{white}{black}{\neg v}{0}{\dia}{\neg v}{\dia}
		\fulltile{4.5}{0}{colAux}{black}{\tl b}{\str}{}{\str}{\dia}
		
		\fulltile{0}{1.5}{white}{black}{\overset0\rightarrow}{\dia}{0}{\dia}{0}
		\fulltile{1.5}{1.5}{white}{black}{\tl{or}}{\tl B}{1}{1}{0}
		\fulltile{3}{1.5}{white}{black}{\tl{or}}{\tl C}{1}{0}{1}
		\fulltile{4.5}{1.5}{colPos}{white}{\pmb+}{\str}{}{\str}{1}

		\node at (2.25,-1) {$r_3(v)$};
	\end{scope}
	
	\end{tikzpicture}	
	\hspace*{-2cm}
	
	\bigskip

	\hspace*{-2cm}
	\begin{tikzpicture}
			[scale=\fullscale,every node/.style={scale=\fullscalenode}]

	\begin{scope}
		\fulltile{0}{0}{colAux}{black}{\tl{Z_4}}{\tri}{\dia}{}{}
		\fulltile{1.5}{0}{white}{black}{u}{0}{\dia}{u}{\dia}

		\fulltile{0}{1.5}{colAux}{black}{u}{}{u}{\tri}{}
		\fulltile{1.5}{1.5}{colMod}{white}{u}{}{}{0}{u}
		
		\node at (.75,-1) {$r_1(u)$};
	\end{scope}

	\begin{scope}[xshift=4cm]
		\fulltile{0}{0}{colAux}{black}{\tl{Z_4}}{\tri}{\dia}{}{}
		\fulltile{1.5}{0}{white}{black}{\neg u}{1}{\dia}{\neg u}{\dia}

		\fulltile{0}{1.5}{colAux}{black}{u}{}{u}{\tri}{}
		\fulltile{1.5}{1.5}{colModd}{white}{\tilde u}{}{}{1}{u}
		
		\node at (.75,-1) {$r_2(u)$};
	\end{scope}
	
	\begin{scope}[xshift=8cm]
		\fulltile{0}{0}{colAux}{black}{\tl a}{\dia}{\dia}{\dia}{\dia}
		\fulltile{1.5}{0}{white}{black}{u}{0}{\dia}{u}{\dia}
		\fulltile{3}{0}{white}{black}{\neg u}{1}{\dia}{\neg u}{\dia}
		\fulltile{4.5}{0}{colAux}{black}{\tl b}{\str}{}{\str}{\dia}
		
		\fulltile{0}{1.5}{white}{black}{\overset0\rightarrow}{\dia}{0}{\dia}{0}
		\fulltile{1.5}{1.5}{white}{black}{\tl{or}}{\tl A}{0}{0}{0}
		\fulltile{3}{1.5}{white}{black}{\tl{or}}{\tl B}{1}{1}{0}
		\fulltile{4.5}{1.5}{colPos}{white}{\pmb+}{\str}{}{\str}{1}

		\node at (2.25,-1) {$r_3(u)$};
	\end{scope}
	\end{tikzpicture}
	\hspace*{-2cm}

	\bigskip
	\begin{tikzpicture}
			[scale=\fullscale,every node/.style={scale=\fullscalenode}]
		\fulltile{0}{0}{colAux}{black}{\tl a}{\dia}{\dia}{\dia}{\dia}
		\fulltile{1.5}{0}{white}{black}{u}{0}{\dia}{u}{\dia}
		\fulltile{3}{0}{white}{black}{v}{1}{\dia}{v}{\dia}
		\fulltile{4.5}{0}{white}{black}{\neg w}{0}{\dia}{\neg w}{\dia}
		\fulltile{6}{0}{colAux}{black}{\tl b}{\str}{}{\str}{\dia}
		
		\fulltile{0}{1.5}{white}{black}{\overset0\rightarrow}{\dia}{0}{\dia}{0}
		\fulltile{1.5}{1.5}{white}{black}{\tl{or}}{\tl A}{0}{0}{0}
		\fulltile{3}{1.5}{white}{black}{\tl{or}}{\tl B}{1}{1}{0}
		\fulltile{4.5}{1.5}{white}{black}{\tl{or}}{\tl C}{1}{0}{1}
		\fulltile{6}{1.5}{colPos}{white}{\pmb+}{\str}{}{\str}{1}

		\node at (3,-1) {$s(C)$};
	\end{tikzpicture}
\end{center}

\section{\texorpdfstring{\NP-hardness of \tmbpats}{NP-hardness of 3-MBPATS}}
\label{sec:bwg}

The purpose of this section is to prove the \NP-hardness of \tmbpats.
Let us define a set of restricted input pairs $\cI$ for \pats.
The set \cI contains all pairs $(P,m)$ where $P = P_F$ is the pattern for a formula $F$ in 3-CNF as defined in Sect.~\ref{sec:pats} and $m = \abs{C(P)} + 3$.
Consider the following restriction of \pats.

\myproblem{A pair $(P,m)$ from $\cI$}
	{``Yes'' if $P$ can uniquely be self-assembled by using $m$ tile types}
	{\modpats}
	
As we choose exactly those pairs $(P,m)$ as input for the problem that are generated by the reduction, stated in Theorem~\ref{thm:reduction:pats}, we obtain the following corollary which forms the foundation for the result in this section.

\begin{corollary}
\modpats is \NP-hard.
\end{corollary}

The \NP-hardness of \tmbpats follows by the polynomial-time reduction from \modpats to \tmbpats, stated in Theorem~\ref{thm:reduction}.

\begin{theorem}\label{thm:reduction}
For every input pair $(P,m)\in \cI$ there exist a black/white/gray-colored pattern $Q$ and integers $m_b, m_w, m_g$ such that:
$P$ can be self-assembled by a \tas with at most $m$ tile types if and only if $Q$ can be self-assembled by a \tas with at most $m_b$ black tile types, $m_w$ white tile types, and $m_g$ gray tile types.
Moreover, the tuple $(Q,m_b,m_w,m_g)$ can be computed from $P$ in polynomial time.
\end{theorem}

Lemma~\ref{lem:if} states the ``if part'' and Lemma~\ref{lem:only-if} states the ``only if part'' of Theorem~\ref{thm:reduction}.
Lemma~\ref{lem:poly:time} states that $(Q,m_b,m_w,m_g)$ can be computed from $P$ in polynomial time.

\begin{corollary}\label{cor:mbpats}
\tmbpats is \NP-hard.
\end{corollary}

For the remainder of this section, let $(P,m)\in\cI$ be one fixed pair, let $C=C(P)$ and $k = \abs C$. 
We may assume that $C=[k]$ is a subset of the positive integers.
The tile bounds are
\begin{inparaitem}[]
	\item $m_b = 1$ for black tile types,
	\item $m_w = 5 k - 3(w(P)+h(P)) + 14$ for white tile types, and
	\item $m_g = 2 k + 3$ for gray tile types.
\end{inparaitem}
Note that, due to the pattern design in Sect.~\ref{sec:pats}, $h(P) = 6$ is constant.

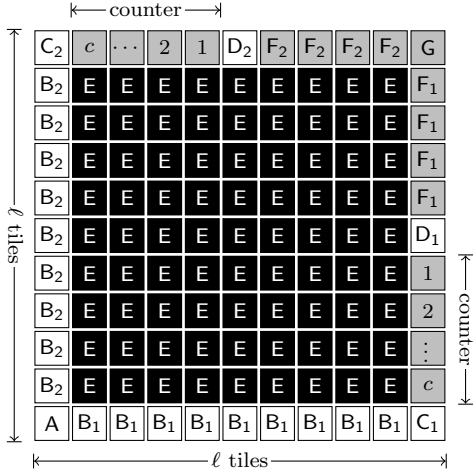
\begin{figure}[ht]
	\centering
	\begin{tikzpicture}[scale=\smallscale]
		\smalltile{0}{0}{white}{\tl{A}}

		\foreach \i in {.5,1,...,4.5}
		{
			\smalltile{\i}{0}{white}{\tl{B_1}}
			\smalltile{0}{\i}{white}{\tl{B_2}}
		}
		\smalltile{5}{0}{white}{\tl{C_1}}
		\smalltile{0}{5}{white}{\tl{C_2}}
		\smalltile{2.5}{5}{white}{\tl{D_2}};
		\smalltile{5}{2.5}{white}{\tl{D_1}};

		\foreach \x in {.5,1,...,4.5}
			\foreach \y in {.5,1,...,4.5}
				\smalltile{\x}{\y}{black}{\color{white}\tl{E}};

		\foreach \i/\l in {.5/c,1/\cdots,1.5/2,2/1}
		{
			\smalltile{\i}{5}{lightgray}{$\l$}
		}
		\foreach \i/\l in {.5/c,1/{\vdots\vphantom{g_g}},1.5/2,2/1}
		{
			\smalltile{5}{\i}{lightgray}{$\l$}
		}
		\foreach \i in {3,3.5,...,4.5}
		{
			\smalltile{\i}{5}{lightgray}{\tl{F_2}}
			\smalltile{5}{\i}{lightgray}{\tl{F_1}}
		}
		
		\smalltile{5}{5}{lightgray}{\tl{G}}

		\begin{scope}[text height=1ex,text depth=0ex,inner sep=1pt,
				font=\footnotesize]
			\draw [|<->|] (.25,5.5) --
				node [fill=white] {counter} (2.25,5.5);
			\draw [|<->|] (5.5,.25) --
				node [fill=white,rotate=-90] {counter} (5.5,2.25);
			\draw [|<->|] (-.25,-.5) --
				node [fill=white] {$\ell$ tiles} (5.25,-.5);
			\draw [|<->|] (-.5,-.25) --
				node [fill=white,rotate=-90] {$\ell$ tiles} (-.5,5.25);
		\end{scope}
	\end{tikzpicture}
	\caption{Black/white/gray supertile which portrays a color $c\in C$.}
	\label{fig:supertile}
\end{figure}

Let $\ell = 5k+8$.
For a color $c\in C$, we define an $\ell\times \ell$ square pattern as shown in Fig.~\ref{fig:supertile}.
We refer to this pattern as well as to its underlying tile assignment as {\em supertile}.
The blowup of such a supertile with a possible tile assignment is shown in Figure~\ref{fig:blowup}.
In contrast to the previous section, the positions in the supertile are labeled which does not mean that the colors or the tiles used to self-assemble the pattern are labeled; the colors are black, white, or gray.
The horizontal and vertical {\em color counters} are the $c$ gray positions in the top row, respectively right column, which are succeeded by a white tile in position \tl{D_2}, respectively \tl{D_1}. 
The color counters illustrate the color $c$ that is {\em portrayed} by the supertile.
The patterns of two supertiles which portray two distinct colors differ only in the place the white tile is positioned in its top row and right column.

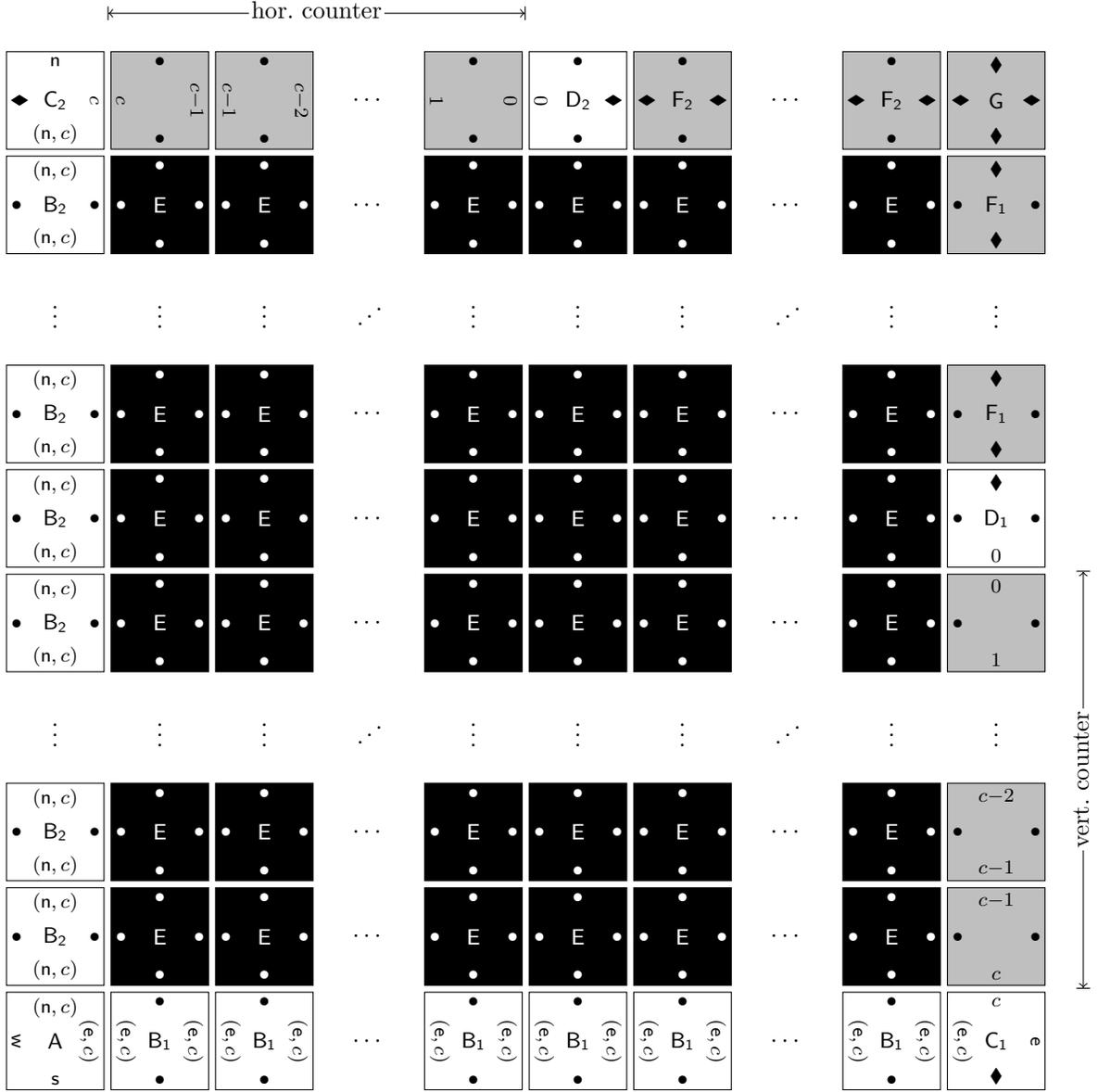
\begin{figure}[p!]
	\centering
	\hspace*{-2cm}%
	\begin{tikzpicture}
			[scale=\fullscale,every node/.style={scale=\fullscalenode}]
		\fulltile{0}{0}{white}{black}{\tl{A}}{(\tl{n},c)}{(\tl{e},c)}{\tl{s}}{\tl{w}}

		\foreach \i in {1.5,3,6,7.5,9,12}
		{
			\fulltile{\i}{0}{white}{black}{\tl{B_1}}{\bul}{(\tl{e},c)}{\bul}{(\tl{e},c)}
			\fulltile{0}{\i}{white}{black}{\tl{B_2}}{(\tl{n},c)}{\bul}{(\tl{n},c)}{\bul}
		}
		\fulltile{13.5}{0}{white}{black}{\tl{C_1}}{c}{\tl{e}}{\dia}{(\tl{e},c)}
		\fulltile{0}{13.5}{white}{black}{\tl{C_2}}{\tl{n}}{c}{(\tl{n},c)}{\dia}
		\fulltile{7.5}{13.5}{white}{black}{\tl{D_2}}{\bul}{\dia}{\bul}{0}
		\fulltile{13.5}{7.5}{white}{black}{\tl{D_1}}{\dia}{\bul}{0}{\bul}

		\foreach \x in {1.5,3,6,7.5,9,12}
			\foreach \y in {1.5,3,6,7.5,9,12}
				\fulltile{\x}{\y}{black}{white}{\tl{E}}{\bul}{\bul}{\bul}{\bul};

		\fulltile{13.5}{1.5}{lightgray}{black}{}{c{-}1}{\bul}{c}{\bul}
		\fulltile{13.5}{3}{lightgray}{black}{}{c{-}2}{\bul}{c{-}1}{\bul}
		\fulltile{13.5}{6}{lightgray}{black}{}{0}{\bul}{1}{\bul}
		\fulltile{13.5}{9}{lightgray}{black}{\tl{F_1}}{\dia}{\bul}{\dia}{\bul}
		\fulltile{13.5}{12}{lightgray}{black}{\tl{F_1}}{\dia}{\bul}{\dia}{\bul}

		\fulltile{1.5}{13.5}{lightgray}{black}{}{\bul}{c{-}1}{\bul}{c}
		\fulltile{3}{13.5}{lightgray}{black}{}{\bul}{c{-}2}{\bul}{c{-}1}
		\fulltile{6}{13.5}{lightgray}{black}{}{\bul}{0}{\bul}{1}
		\fulltile{9}{13.5}{lightgray}{black}{\tl{F_2}}{\bul}{\dia}{\bul}{\dia}
		\fulltile{12}{13.5}{lightgray}{black}{\tl{F_2}}{\bul}{\dia}{\bul}{\dia}

		\fulltile{13.5}{13.5}{lightgray}{black}{\tl{G}}{\dia}{\dia}{\dia}{\dia}

		\begin{scope}[text height=1ex,text depth=0ex,inner sep=1pt]
			\draw [|<->|] (.75,14.75) --
				node [fill=white] {hor.~counter} (6.75,14.75);
			\draw [|<->|] (14.75,.75) --
				node [fill=white,rotate=90] {vert.~counter} (14.75,6.75);
		\end{scope}
		
		\foreach \i in {0,1.5,3,6,7.5,9,12,13.5}
		{
			\node at (\i,4.5) {$\vdots$};
			\node at (\i,10.5) {$\vdots$};
			\node at (4.5,\i) {$\cdots$};
			\node at (10.5,\i) {$\cdots$};
		}
		\node at (4.5,4.5) {$\iddots$};
		\node at (4.5,10.5) {$\iddots$};
		\node at (10.5,4.5) {$\iddots$};
		\node at (10.5,10.5) {$\iddots$};
	\end{tikzpicture}%
	\hspace*{-2cm}
\caption{
Blowup of one supertile with possible tile assignment, representing a tile with color $c$ and glues $\tl{n}$, $\tl{e}$, $\tl{s}$, and $\tl{w}$ on its north, east, south, and west edges, respectively.}
\label{fig:blowup}
\end{figure}

For colors in the bottom row and left column of the pattern $P$ we use {\em incomplete supertiles}:
a supertile portraying a color $c$ in the bottom row of pattern $P$ lacks the white row with positions \tl A, \tl{B_1}, and \tl{C_1};
a supertile representing a color $c$ in the left column of pattern $P$ lacks the white column with positions \tl A, \tl{B_2}, and \tl{C_2}.
In particular, the supertile portraying color $P(1,1)$ does not contain any of the positions \tl A, \tl{B_1}, \tl{B_2}, \tl{C_1}, and \tl{C_2}.
Recall that all incomplete supertiles portray a color $c$ that is unique in $P$.

\begin{figure}[th]
	\centering
	\begin{tikzpicture}
		\fill [lightgray] (10,-.85) rectangle (10.45,3);
		\fill [black] (10,-.85) rectangle (10.15,.85);
		\fill [black] (10,1) rectangle (10.15,2.85);
		\fill [black] (10.3,-.85) rectangle (10.45,.85);
		\fill [black] (10.3,1) rectangle (10.45,2.85);
		\fill [white] (10.15,.15*6-.85) rectangle (10.3,.15*6-.7);

		\fill [lightgray] (-.85,6) rectangle (5,6.45);
		\fill [black] (-.85,6) rectangle (.85,6.15);
		\fill [black] (1,6) rectangle (2.85,6.15);
		\fill [black] (3,6) rectangle (4.85,6.15);
		\fill [black] (-.85,6.3) rectangle (.85,6.45);
		\fill [black] (1,6.3) rectangle (2.85,6.45);
		\fill [black] (3,6.3) rectangle (4.85,6.45);
		\fill [white] (.15*6-.85,6.15) rectangle (.15*6-.7,6.3);

		\fill [lightgray] (6,6) rectangle (10.45,6.45);
		\fill [lightgray] (10,4) rectangle (10.45,6.45);
		\fill [black] (10,4) rectangle (10.15,5.85);
		\fill [black] (10.3,4) rectangle (10.45,5.85);
		\fill [black] (6,6) rectangle (7.85,6.15);
		\fill [black] (8,6) rectangle (9.85,6.15);
		\fill [black] (6,6.3) rectangle (7.85,6.45);
		\fill [black] (8,6.3) rectangle (9.85,6.45);
		\fill [black] (10,6) rectangle (10.15,6.15);
		\fill [black] (10,6.3) rectangle (10.15,6.45);
		\fill [black] (10.3,6) rectangle (10.45,6.15);
		\fill [black] (10.3,6.3) rectangle (10.45,6.45);
		

		\supertileBL{0}{0}{1}{supertile portraying $P(1,1)$}
		\supertileB{2}{0}{2}{supertile portraying $P(2,1)$}
		\supertileB{4}{0}{3}{supertile portraying $P(3,1)$}
		\supertileL{0}{2}{2}{supertile portraying $P(1,2)$}
		\supertile{2}{2}{4}{supertile portraying $P(2,2)$}
		\supertile{4}{2}{3}{supertile portraying $P(3,2)$}

		\supertileB{7}{0}{3}{supertile portraying $P(\tl w{-}1,1)$}
		\supertileB{9}{0}{3}{supertile portraying $P(\tl w,1)$}
		\supertile{7}{2}{5}{supertile portraying $P(\tl w{-}1,2)$}
		\supertile{9}{2}{6}{supertile portraying $P(\tl w,2)$}

		\supertileL{0}{5}{2}{supertile portraying $P(1,\tl h)$}
		\supertile{2}{5}{6}{supertile portraying $P(2,\tl h)$}
		\supertile{4}{5}{5}{supertile portraying $P(3,\tl h)$}

		\supertile{7}{5}{1}{supertile portraying $P(\tl w{-}1,\tl h)$}
		\supertile{9}{5}{3}{supertile portraying $P(\tl w,\tl h)$}
		
		\draw[very thick,dotted] (5.25,0) -- (5.75,0);
		\draw[very thick,dotted] (5.25,2) -- (5.75,2);
		\draw[very thick,dotted] (5.25,5) -- (5.75,5);

		\draw[very thick,dotted] (0,3.25) -- (0,3.75);
		\draw[very thick,dotted] (2,3.25) -- (2,3.75);
		\draw[very thick,dotted] (4,3.25) -- (4,3.75);
		\draw[very thick,dotted] (7,3.25) -- (7,3.75);
		\draw[very thick,dotted] (9,3.25) -- (9,3.75);

		\draw[very thick,dotted] (5.3,3.3) -- (5.7,3.7);
	\end{tikzpicture}
	\caption{Black/white/gray pattern $Q$ defined by the $k$-color pattern $P$ with $\tl w = w(P)$ and $\tl h=h(P)$.}
	\label{fig:patternQ}
\end{figure}
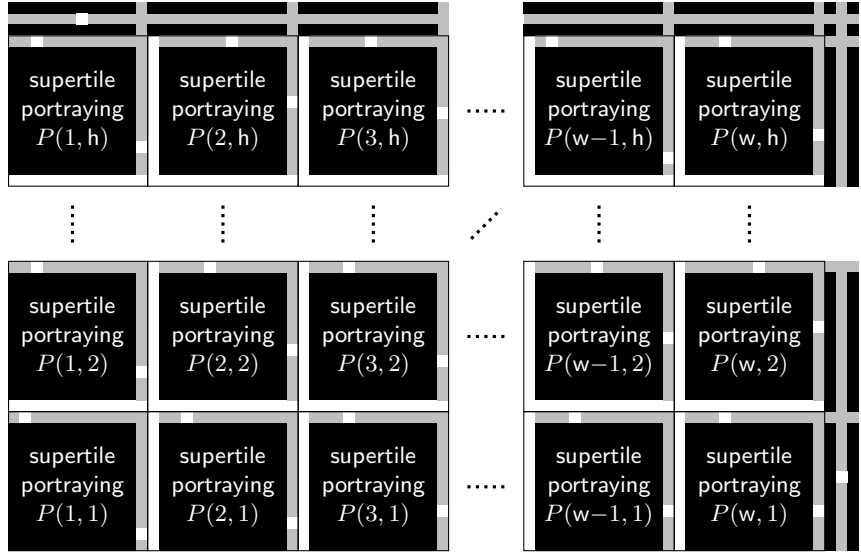

The pattern $Q$ is shown in Fig.~\ref{fig:patternQ}.
By $Q\gen{x,y}$ we denote the pattern of the supertile covering the square area spanned by positions $((x-1)\cdot \ell,(y-1)\cdot \ell)$ and $(x\cdot \ell-1,y\cdot \ell-1)$ in $Q$;
the incomplete supertiles cover one row and/or column less.
The pattern is designed such that supertile $Q\gen{x,y}$ portrays the color $P(x,y)$ for all $x\in[w(P)]$ and $y\in[h(P)]$.
Additionally, $Q$ contains three {\em gadget rows} and three {\em gadget columns} which are explained in Fig.~\ref{fig:gadget}.
The purpose of these gadget rows and columns is to ensure that the color counters can only be implemented in one way when using no more than $m_g$ gray tile types.
All together $Q$ is of dimensions $w(Q) = \ell\cdot w(P) +2$ times $h(Q) = \ell\cdot h(P) +2$.
Obviously, the pattern $Q$ can be computed from $P$ in polynomial time.

\begin{lemma}\label{lem:poly:time}
$(Q,m_b,m_w,m_g)$ can be computed from $P$ in polynomial time.
\end{lemma}

\begin{proof}
This is obvious by the design of the pattern.
\end{proof}

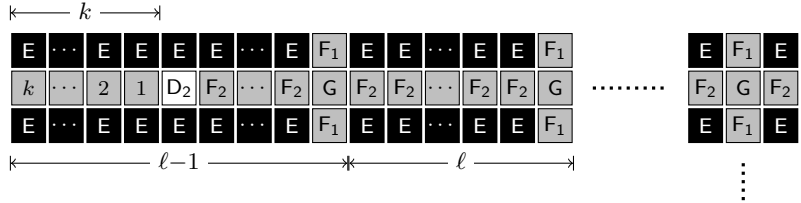
\begin{figure}[th]
	\centering
	\begin{tikzpicture}[scale=\smallscale]
		\foreach \i in {0,1,1.5,2,2.5,3.5,4.5,5,6,6.5,9,10}
		{
			\smalltile{\i}{-.5}{black}{\color{white}\tl{E}}
			\smalltile{\i}{.5}{black}{\color{white}\tl{E}}
		}
		\foreach \i in {.5,3,5.5}
		{
			\smalltile{\i}{-.5}{black}{\color{white}$\cdots$}
			\smalltile{\i}{.5}{black}{\color{white}$\cdots$}
		}

		\foreach \i/\l in {0/k,.5/{\cdots},1/2,1.5/1,
			2.5/\tl{F_2},3/\cdots,3.5/\tl{F_2},4/\tl{G},
			4.5/\tl{F_2},5/\tl{F_2},5.5/\cdots,6/\tl{F_2},6.5/\tl{F_2},7/\tl{G},
			9/\tl{F_2},9.5/\tl{G},10/\tl{F_2}}
		{
			\smalltile{\i}{0}{lightgray}{$\l$}
		}
		\foreach \i in {4,7,9.5}
		{
			\smalltile{\i}{-.5}{lightgray}{\tl{F_1}}
			\smalltile{\i}{.5}{lightgray}{\tl{F_1}}
		}
		
		\smalltile{2}{0}{white}{\tl{D_2}};

		\begin{scope}[text height=1ex,text depth=0ex,inner sep=3pt,
				font=\small]
			\draw [|<->|] (-.25,1) --
				node [fill=white] {$k$} (1.75,1);
			\draw [|<->|] (-.25,-1) --
				node [fill=white] {$\ell{-}1$} (4.25,-1);
			\draw [|<->|] (4.25,-1) --
				node [fill=white] {$\ell$} (7.25,-1);
		\end{scope}
		
		\draw[very thick,dotted] (7.5,0) -- (8.5,0);
		\draw[very thick,dotted] (9.5,-1) -- (9.5,-1.5);
	\end{tikzpicture}
	\caption{The gadget rows on the north border of the pattern $Q$, the gadget columns are symmetrical:
	the middle row (\resp column) contains gray tiles except for one white tile in position $k+1$;
	the upper and lower rows (\resp left and right columns) contain gray tiles in positions above the gray column (\resp right of the gray row) of a supertile, the other tiles are black.}
	\label{fig:gadget}
\end{figure}

For a \tas $\Theta$ which self-assembles $Q$, we extend our previous notion such that $Q_\Theta\gen{x,y}$ denotes the tile assignment of supertile $Q\gen{x,y}$ given by $\Theta$.
In the following, we will prove properties of such a \tas $\Theta$.
Our first observation is about the black and gray tile types plus two of the white tile types.

\begin{figure}[th]
	\centering
	\begin{tikzpicture}
			[scale=\fullscale,every node/.style={scale=\fullscalenode}]
		\fulltile{0}{0}{black}{white}{\tl E}{\bullet}{\bullet}{\bullet}{\bullet}
		\fulltile{3.5}{0}{white}{black}{\tl{D_1}}%
			{\blacklozenge}{\bullet}{0}{\bullet}
		\fulltile{5.25}{0}{white}{black}{\tl{D_2}}%
			{\bullet}{\blacklozenge}{\bullet}{0}

	\begin{scope}[yshift=-1.75cm]
		\fulltile{0}{0}{lightgray}{black}{}{i{-}1}{\bullet}{i}{\bullet}
		\fulltile{1.75}{0}{lightgray}{black}{}{\bullet}{i{-}1}{\bullet}{i}
		\fulltile{3.5}{0}{lightgray}{black}{\tl{F_1}}%
			{\blacklozenge}{\bullet}{\blacklozenge}{\bullet}
		\fulltile{5.25}{0}{lightgray}{black}{\tl{F_2}}%
			{\bullet}{\blacklozenge}{\bullet}{\blacklozenge}
		\fulltile{7}{0}{lightgray}{black}{\tl G}%
			{\blacklozenge}{\blacklozenge}{\blacklozenge}{\blacklozenge}
	\end{scope}
	\end{tikzpicture}
	\caption{The black tile type, two of the white tile types, and all gray tile types: the labeled tile types are used in the corresponding positions of each supertile and the gadget pattern; the unlabeled tile types, called {\em counter tiles} for $i\in[k]$, implement the vertical and horizontal color counters.}
	\label{fig:black:gray}
\end{figure}
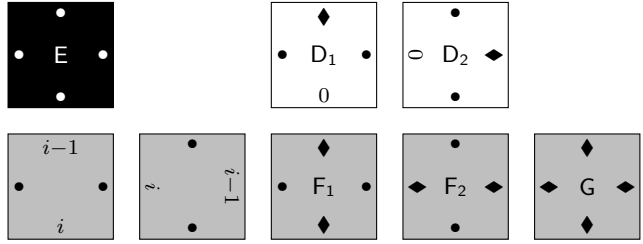

\begin{lemma}\label{lem:black:gray}
Let $\Theta$ be a \tas which self-assembles the pattern $Q$ using at most $m_b=1$ black tile types and $m_g=2k+3$ gray tile types.
The black and gray tile types in $\Theta$ are of the form shown in Fig.~\ref{fig:black:gray} and $\Theta$ contains two white tiles of the form shown in the figure.
In every supertile, the horizontal and vertical color counters are implemented by a subset of the counter tile types and for a position \tl{E}, \tl{D_1}, \tl{D_2}, \tl{F_1}, \tl{F_2}, or \tl G the correspondingly labeled tile type is used.
Furthermore, the glues $\bul,\dia,0,1,\ldots, k$ are all distinct.
\end{lemma}

\begin{proof}
In every supertile we find a black square that consists of more than just one tile, therefore, the sole black tile type must have the same north and south glues, respectively east and west glues.
We may assume that \bul is the glue on all four edges of the black tile type.

Now, we have a look at the tile assignment of the three gadget rows; see Fig.~\ref{fig:gadget}.
In the middle row, the first $k$ gray tiles are succeeded by a white tile.
As the south input of all $k$ tiles is the north glue \bul of the black tiles, if the same tile type would be used in two of these $k$ positions, there were a cycle in the gray tile types used which would be repeated over and over and no white tile could occur after $k$ gray tiles; furthermore, all the east outputs (\resp west inputs) of these tiles are distinct.
Thus, these tiles implement a horizontal counter which is capable of counting to $k$, or more intuitively, counting from $k$ downto $1$.
Next a white tile is used whose west input is distinct from its east output; otherwise, the same white tile type would be used again.
As the white tile is succeeded by $4k+6$ gray tiles with south input \bul, all gray tiles used in this row have to occur at least twice, therefore, all tiles, used in these positions, have to be distinct from the tiles implementing the horizontal counter; we use the label \dia for the east output of the white tile in position \tl{D_2} and the west input of the gray tile in left-most position \tl{F_2}.
Finally, observe that all tiles described so far have south and north glue \bul and that their east and west glues are distinct from \bul.

From the three gadget columns we obtain analogous results for the vertical counter, the white tile in position \tl{D_1}, and the gray tiles in the first $4k+6$ positions \tl{F_1}.
None of the tiles in the middle column can coincide with a tile in the middle row since they have the glue \bul on different inputs.

Since we need at least one tile for each of the positions \tl{F_1} and \tl{F_2}, there is only one gray tile type left in $\Theta$.
Any tile at position \tl G has north and east outputs from $\dia, 0,1,\ldots, k$ as these edges are adjacent to gray tiles.
We cannot use a tile which we have described so far for any of the positions \tl G.
Therefore, all tiles in positions \tl{F_1} (\resp \tl{F_2}) in the gadget are of the same type whose north and south (\resp east and west) glues equal.
All positions \tl G share the same tile type $\tau$ whose south and west inputs are \dia.
If the east output of $\tau$ were one of $0,1,\ldots, k$, then a white tile were among the $k+1$ tiles succeeding the left-most position \tl G in the middle gadget row; hence, $\tau(E) = \dia$ and, symmetrically, $\tau(N) = \dia$.

All supertiles have to share the same black and gray tile types as there are no other black and gray tile types in $\Theta$.
The color counters have to be implemented by the counter tile types.
As the south and west inputs of the tiles in positions \tl{D_1} and \tl{D_2} are determined by its gray and black neighbors, it is also clear that only the described white tiles can be used in these positions.
Now, the types of the tiles in positions \tl{F_1}, \tl{F_2}, and \tl G are also determined by their inputs.
\end{proof}

\begin{remark}\label{rem:controltile}
Consider a \tas $\Theta$ that self-assembles the pattern $Q$ using most $m_b$ black tile types and $m_g$ gray tiletypes.
If we have a look at the tile assignment of the black square plus the gray column and row in a supertile, we see that this block has inputs \bul on all edges except for edges where the color counters are initialized and it has outputs \bul on all edges, except for its right-most and top-most output edges which are \dia.
This means that all information on how to initialize the color counters has to be carried through the white lines and rows, that are, the tiles in positions \tl A, \tl{B_1}, \tl{B_2}, \tl{C_1}, \tl{C_2}.
Moreover, the tile in position \tl A is the only one with non-generic input from other supertiles.
This tile fully determines the tile assignment of the supertile and can be seen as the {\em control tile} or {\em seed} of the supertile.
Henceforth, for a supertile $s=Q_\Theta\gen{x,y}$ we extend our notion of glues such that $s(S)$ and $s(W)$ denote the south and west input of the tile in position \tl A, respectively, $s(N)$ and $s(E)$ denote the north and east output of the tiles in positions \tl{C_2} and \tl{C_1}, respectively.
For incomplete supertiles only one of $s(N)$ or $s(E)$ is defined.

Two supertiles in $Q_\Theta$ are considered distinct if their tile assignment differs in at least one position.
By the observations above, two complete supertiles are distinct if and only if their control tiles are of distinct types; this is equivalent to require that  the inputs of the two supertiles differ. 
Since incomplete supertiles portray unique colors in $P$, they are distinct from any supertile in $Q_\Theta$ but itself.
\end{remark}

There is some flexibility in how the white tile types are implemented in a \tas $\Theta$ which self-assembles $Q$.
Let us present one possibility which proves the ``only if part'' of Theorem~\ref{thm:reduction}.

\begin{lemma}\label{lem:only-if}
If $P$ can be self-assembled by a \tas $T$ with $m$ tile types, then $Q$ can be self-assembled by a \tas $\Theta$ using $m_b$ black tile types, $m_w$ white tile types, and $m_g$ gray tile types.
\end{lemma}

\begin{proof}
Let the \tas $\Theta$ contain the tile types given in Fig.~\ref{fig:black:gray}.
By Lemma~\ref{lem:black:gray}, the tiles of these types can self-assemble the gadget pattern and the black square plus color counters in every supertile.
Recall from Lemma~\ref{lem:q_or} that $T$ contains four tile types for the \tl{or}-gate and only one tile type for every other color in $C$.

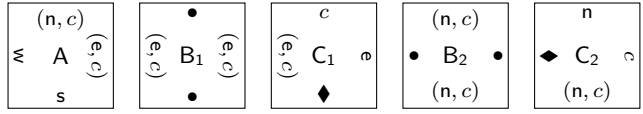
\begin{figure}[th]
	\centering
	\begin{tikzpicture}
			[scale=\fullscale,every node/.style={scale=\fullscalenode}]
		\fulltile{0}{0}{white}{black}{\tl A}{(\tl n,c)}{(\tl e,c)}{\tl s}{\tl w}
		\fulltile{1.75}{0}{white}{black}{\tl{B_1}}{\bul}{(\tl e,c)}{\bul}{(\tl e,c)}
		\fulltile{3.5}{0}{white}{black}{\tl{C_1}}{c}{\tl e}{\dia}{(\tl e,c)}
		\fulltile{5.25}{0}{white}{black}{\tl{B_2}}{(\tl n,c)}{\bul}{(\tl n,c)}{\bul}
		\fulltile{7}{0}{white}{black}{\tl{C_2}}{\tl n}{c}{(\tl n,c)}{\dia}
	\end{tikzpicture}
	\caption{White tile types for the supertile portraying a color $c\in C$, except for the \tl{or}-gate, where $t\in T$ with $c=\chi(t)$, $\tl n =t(N)$, $\tl e = t(E)$, $\tl s=t(S)$, and $\tl w = t(W)$.}
	\label{fig:white1}
\end{figure}

First, consider $c\in C\sm\set{\Cor}$ and let $t\in T$ be the tile type with color $c$.
We use the tile types in Fig.~\ref{fig:white1} for the supertile portraying $c$.
Note that none of the five tile types share the same inputs.
The labels on the them depict the positions in the supertile where each tile is used.
We do not need any of these white tile types for the incomplete supertile representing the color $P(1,1)$.
For the other incomplete supertiles we only need two of these tile types.
Recall that, by design, a color which is portrayed by an incomplete supertile in $Q$ is unique in $P$.
This amounts to $5 \cdot (\abs{C}-1) - 3 \cdot ( h(P)+ w(P)) + 1$ white tile types in $\Theta$ for all the supertiles portraying colors in $C\sm\set{\Cor}$.
We have one remark for the colors in the top row and right column of $P$:
since these colors do not occur in any other position of the pattern and the north and/or east glues of the respective tiles are not used in $P$, we may assume that all these glues are \bul; this allows for the proper attachment of the gadget pattern.

\begin{figure}[th]
	\centering
	\begin{tikzpicture}
			[scale=\fullscale,every node/.style={scale=\fullscalenode}]
		\node [anchor=east] at (-1,0) {$o_1$:};
		\fulltile{0}{0}{white}{black}{\tl A}{(\tl A,c)}{(0,c)}{0}{0}
		\fulltile{1.75}{0}{white}{black}{\tl{B_2}}{(\tl A,c)}{\bul}{(\tl A,c)}{\bul}
		\fulltile{3.5}{0}{white}{black}{\tl{C_2}}{\tl A}{c}{(\tl A,c)}{\dia}

		\fulltile{5.5}{0}{white}{black}{\tl{B_1}}{\bul}{(0,c)}{\bul}{(0,c)}
		\fulltile{7.25}{0}{white}{black}{\tl{C_1}}{c}{0}{\dia}{(0,c)}

	\begin{scope}[yshift=-1.75cm]
		\node [anchor=east] at (-1,0) {$o_2$:};
		\fulltile{0}{0}{white}{black}{\tl A}{(\tl B,c)}{(1,c)}{1}{0}
		\fulltile{1.75}{0}{white}{black}{\tl{B_2}}{(\tl B,c)}{\bul}{(\tl B,c)}{\bul}
		\fulltile{3.5}{0}{white}{black}{\tl{C_2}}{\tl B}{c}{(\tl B,c)}{\dia}
	\end{scope}

	\begin{scope}[yshift=-3.5cm]
		\node [anchor=east] at (-1,0) {$o_3$:};
		\fulltile{0}{0}{white}{black}{\tl A}{(\tl C,c)}{(1,c)}{0}{1}
		\fulltile{1.75}{0}{white}{black}{\tl{B_2}}{(\tl C,c)}{\bul}{(\tl C,c)}{\bul}
		\fulltile{3.5}{0}{white}{black}{\tl{C_2}}{\tl C}{c}{(\tl C,c)}{\dia}
		\fulltile{5.5}{0}{white}{black}{\tl{B_1}}{\bul}{(1,c)}{\bul}{(1,c)}
		\fulltile{7.25}{0}{white}{black}{\tl{C_1}}{c}{1}{\dia}{(1,c)}
	\end{scope}

	\begin{scope}[yshift=-5.25cm]
		\node [anchor=east] at (-1,0) {$o_4$:};
		\fulltile{0}{0}{white}{black}{\tl A}{(\tl D,c)}{(1,c)}{1}{1}
		\fulltile{1.75}{0}{white}{black}{\tl{B_2}}{(\tl D,c)}{\bul}{(\tl D,c)}{\bul}
		\fulltile{3.5}{0}{white}{black}{\tl{C_2}}{\tl D}{c}{(\tl D,c)}{\dia}
	\end{scope}
	
	\draw [decoration={brace,amplitude=\fullscale*6pt},decorate] (4.5cm-3pt,-1cm) -- (4.5cm-\fullscale*3pt,-6cm);

	\end{tikzpicture}
	\caption{White tile types for supertiles portraying the \tl{or}-gate	where $o_1,o_2,o_3,o_4\in T$ are defined in Fig.~\ref{fig:or-tiles}.}
	\label{fig:whiteor}
\end{figure}
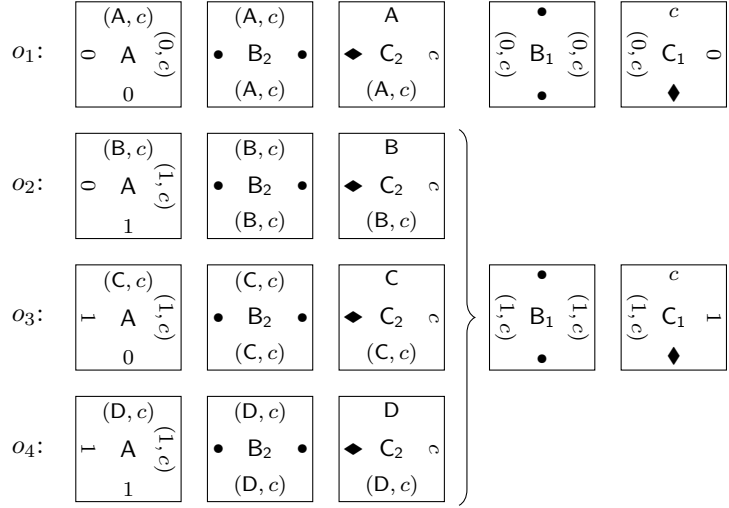

Now, consider $c = \Cor$.
Recall from Lemma~\ref{lem:q_or} that there are $o_1,o_2,o_3,o_4\in T$ with color $c$ as depicted in Fig.~\ref{fig:or-tiles}.
We use the 16 white tile types in Fig.~\ref{fig:whiteor} for the four supertiles portraying $c$.
The supertiles for $o_2$, $o_3$, and $o_4$ share tile types in positions \tl{B_1} and \tl{C_1}, as labelled.
Note that the inputs of the 16 tile types are mutually distinct.

Let $\Theta$ contain all the white tile types we have defined plus the two white tile types from Fig.~\ref{fig:black:gray} and note that the white tile types add up to 
\[
	m_w = 5 \cdot \abs{C} - 3 \cdot ( h(P)+ w(P)) + 14
\]
as desired.
If two distinct tile types in $\Theta$ had the same inputs, it had to be two tile types for position \tl A which implies that two distinct tile types in $T$ would have same inputs as well; thus, $\Theta$ is a properly defined \tas.
The $L$-shaped seed of $\Theta$ is defined such that the incomplete supertiles for the bottom row and left column, as well as the gadgets, can properly attach.

By Remark~\ref{rem:controltile}, it is clear that, starting from its control tile (or the $L$-shaped seed), every supertile properly self-assembles as long as its west and south neighboring supertiles are present.
Let $Q'$ be the pattern that is self-assembled by $\Theta$.
Using induction over $x$ and $y$, the supertile $s=Q'_\Theta\gen{x,y}$ represents the tile $t=P_T(x,y)$ because the respective glues of $s$ and $t$ coincide, for all $x\in[w(P)]$ and $y\in[h(P)]$.
Furthermore, by design of the supertiles in the top row and right column of $Q'$, the gadget rows and columns can self-assemble.
We conclude that $\Theta$ self-assembles $Q = Q'$.
\end{proof}

For the converse implication of Theorem~\ref{thm:reduction}, let us show how to obtain a \tas that self-assembles $P$ from the supertiles in $Q_\Theta$.
The following result follows from the bijection between supertiles in $Q_\Theta$ and tiles in $P_T$.

\begin{lemma}\label{lem:st:to:t}
Let $\Theta$ be a \tas which self-assembles $Q$ using at most $m_b$ black tile types and $m_g$ gray tile types, and let 
\[
	S=\sett{Q_\Theta\gen{x,y}}{x\in[w(P)],y\in[h(P)]}
\]
be the set of all distinct supertiles in $Q_\Theta$.
There exists a \tas $T$ with $\abs S$ tile types which self-assembles $P$ such that for each supertile $s\in S$ there exists a tile type $t_s\in T$ with the same glues on the respective edges and $s$ portrays the color of $t_s$.
For an incomplete supertile the statement holds for the defined glue.
\end{lemma}

\begin{proof}
Note that, except for the tiles with colors in the first row and column of $P$, the tiles in $T$ are fully defined. 
For the undefined glues on tile types, representing incomplete supertiles, we introduce unique matching glues; the $L$-shaped seed of $T$ is defined to match these glues.
Clearly, the first column and row of $P$ can be self-assembled by $T$.
Recall from Remark~\ref{rem:controltile} that the control tile in position \tl A of a complete supertile $s$ fully determines the supertile and its outputs.
The placement of the control tile is determined by the east output of the west neighbor and the north output of the south neighbor of $s$.
Let $P'$ be the pattern which is self-assembled by $T$.
Using induction over $x$ and $y$, we see that if $s = Q_\Theta\gen{x,y}$, then $t_s = P'_T(x,y)$ for all $x\in[w(P)]$ and $y\in[h(P)]$;
thus, $T$ self-assembles $P=P'$.
\end{proof}

We continue with the investigation of the white tile types that are used to self-assemble the pattern $Q$.
The next lemma follows by a case study of what would go wrong if one tile type were used in two of the positions.

\begin{lemma}\label{lem:white:positions}
Let $\Theta$ be a \tas which self-assembles the pattern $Q$ using at most $m_b$ black tile types and $m_g$ gray tile types.
A white tile type from $\Theta$ which is used in one of the positions \tl A, \tl{B_1}, \tl{B_2}, \tl{C_1}, \tl{C_2}, \tl{D_1}, or \tl{D_2} cannot be used in another position in any supertile.
\end{lemma}

\begin{proof}
Clearly, we do not have to argue about positions which are gray or black.
By Lemma~\ref{lem:black:gray} and the design of supertiles, a tile type used in a position
\begin{compactenum}[1.)]
	\item \tl{B_1} has south and north glue \bul;
	\item \tl{B_2} has west and east glue \bul,
	\item \tl{C_1} has south glue \dia and north glue in $[k]$;
	\item \tl{C_2} has west glue \dia and east glue in $[k]$;
	\item \tl{D_1} has east and west glue \bul, south glue $0$, and north glue \dia;
	\item \tl{D_2} has south and north glue \bul, west glue $0$, and east glue \dia.
\end{compactenum}

First, suppose the same tile type would be used in position $\tl A$ and one of the positions \tl{B_1}, \tl{B_2}, \tl{C_1}, \tl{C_2}, \tl{D_1}, or \tl{D_2}.
The tile type had north or east output from $\bul, \dia, 1,\ldots,k$; therefore, the north or east neighbor of the tile of this type in position \tl A would be black or gray --- a contradiction.

Due to the distinction in their inputs, the same tile type cannot be used in the following pairs of positions:
$(\tl{B_1},\tl{C_1})$, $(\tl{B_1},\tl{D_1})$,
$(\tl{B_2},\tl{C_2})$, $(\tl{B_2},\tl{D_2})$, 
$(\tl{C_1},\tl{D_1})$, $(\tl{C_1},\tl{D_2})$, 
$(\tl{C_2},\tl{D_1})$, $(\tl{C_2},\tl{D_2})$, 
$(\tl{D_1},\tl{D_2})$.

If a tile type were used in any pair of positions 
$(\tl{B_1},\tl{B_2})$, $(\tl{B_1},\tl{C_2})$, 
$(\tl{B_2},\tl{C_1})$, or 
$(\tl{C_1},\tl{C_2})$
it would have the same inputs as one of the gray or black tiles.

If a tile type were used in position $\tl{B_1}$ and $\tl{D_2}$, then the east neighbor of the tile of this type in position \tl{B_1} would be the gray tile labelled $F_2$ or $G$ from Fig.~\ref{fig:black:gray}.
Symmetrically, no tile type can be used in both positions \tl{B_2} and \tl{D_1}.
\end{proof}

Let \tl{B_1^*} be the right-most position \tl{B_1} in a supertile, adjacent to position \tl{C_1}, and let \tl{B_2^*} be the top-most position \tl{B_2} in a supertile, adjacent to position \tl{C_2}.
The following argument is about tiles in the five positions $K = \set{\tl A, \tl{B_1^*}, \tl{B_2^*}, \tl{C_1}, \tl{C_2}}$ of each supertile.
Following Remark~\ref{rem:controltile} it is clear that a tile in position~\tl{A} fully determines the supertile, tiles in positions \tl{B_1^*} and \tl{C_1} carry the color and the east glue of a supertile, whereas tiles in positions \tl{B_2^*} and \tl{C_2} carry the color and the north glue.

\begin{lemma}\label{lem:white}
Let $\Theta$ be a \tas which self-assembles $Q$ using at most $m_b$ black tile types and $m_g$ gray tile types.
Let $s_1$ and $s_2$ be supertiles in $Q_\Theta$.
\begin{compactenum}[\ i.)]
	\item If $s_1$ and $s_2$ portray different colors, they cannot share any tile types in positions from $K$.
	\item If $s_1(E) \neq s_2(E)$, they cannot share any tile types in \tl A, \tl{B_1^*}, or \tl{C_1}.
	\item If $s_1(N) \neq s_2(N)$, they cannot share any tile types in \tl A, \tl{B_2^*}, or \tl{C_2}.
\end{compactenum}
The three statements hold for all available positions in incomplete supertiles.
\end{lemma}

\begin{proof}
By Lemma~\ref{lem:white:positions}, we do not have to consider mixups of positions.
Firstly, recall from Remark~\ref{rem:controltile} that the tile in position \tl A determines the supertile and its outputs.
Two supertiles portraying different colors or having different outputs cannot share the same tile type in positions \tl A.

Now, consider a supertile $s_1$ representing the color $c$.
The type $\gamma$ of the tile in position \tl{C_1} defines the east output $s_1(E)=\gamma(E)$ and initializes the vertical color counter.
If $\gamma$ is used in another supertile $s_2$ in position \tl{C_1}, then $s_2$ portrays the same color $c$ and $s_1(E) = s_2(E)$.
As the south input of every tile in a position \tl{C_1} is \dia, the type $\beta$ of the tile in position \tl{B_1^*} of $s_1$ determines the placement of $\gamma$ in position \tl{C_1}.
Thus, if $\beta$ is used in another supertile $s_2$ in position \tl{B_1^*}, then $s_2$ portrays the same color $c$ and $s_1(E) = s_2(E)$.
This concludes the proof of statement {\it ii.)}.
Statements~{\it i.)}\ and~{\it iii.)}\ follow by symmetric arguments on the tile types in positions \tl{C_2} an \tl{B_2^*}.
\end{proof}

\begin{remark}
We do not claim that two supertiles could not share any tile types in positions \tl{B_1} or \tl{B_2} while portraying different colors or having different output.
Indeed, consider two supertiles $s_1$ and $s_2$ portraying different colors and let $\beta$ and $\beta'$ be tile types for positions \tl{B_1} with $\beta(E) = \beta'(W) = \tl x$ and $\beta(W) = \beta'(E) = \tl y$.
If the control tile of $s_1$ has east glue \tl x while the control tile of $s_2$ has east glue \tl y, these supertiles have different tile types in positions \tl{C_1} but share the tile types $\beta$ and $\beta'$ in their positions \tl{B_1}.
\end{remark}

Let us conclude the proof of Theorem~\ref{thm:reduction}.

\begin{lemma}\label{lem:if}
The pattern $P$ can be self-assembled by a \tas $T$ with $m$ tile types if $Q$ can be self-assembled by a \tas $\Theta$ with $m_b$ black tile types, $m_w$ white tile types, and $m_g$ gray tile types.
\end{lemma}

\begin{proof}
We will prove that $Q_\Theta$ cannot contain more than $m$ distinct supertiles while $\Theta$ respects the given tile bounds.
Then, the claim follows from Lemma~\ref{lem:st:to:t}.
The black, gray, and two white tile types in $\Theta$ are defined by Lemma~\ref{lem:black:gray}.
We are now counting the minimal number of white tile types that we need to self-assemble the pattern $Q$.
The number of distinct tile types used as control tiles, in positions \tl A, equals to the number of distinct complete supertiles of $Q_\Theta$.

Consider a color $c\in C$.
There is at least one supertile in $Q_\Theta$ which portrays color $c$ and, assuming the supertile is complete, we need five white tile types in positions from $K$ of the supertile, by Lemma~\ref{lem:white:positions}, and these five tile types cannot be used for any of positions from $K$ of a supertile portraying another color, by Lemma~\ref{lem:white}.
If the supertile is incomplete but $c\neq P(1,1)$, hence $c$ is unique in $P$, we only need two white tile types which cannot be used in any of the positions from $K$ of another supertile.
We do not need any additional white tile types for the supertile portraying $P(1,1)$.
For the supertiles portraying colors in $C\sm\set{\Cor}$ we need 
	$5 \cdot (\abs{C} -1) - 3 \cdot (w(P) + h(P)) + 1$
white tiles which cannot be used in a position from $K$ in a supertile representing the \tl{or}-gate;
furthermore, the two white tile types for positions \tl{D_1} and \tl{D_2} also cannot be used in a position from $K$.
Among these white tile types we find only $\abs{C} - w(P) - h(P))$ types used as control tiles.

We only have 16 white tile types left for the \tl{or}-gate supertiles.
From Lemma~\ref{lem:p_or} and Lemma~\ref{lem:st:to:t} we infer that $Q_\Theta$ contains either
\begin{compactenum}[\it i.)]
	\item three distinct supertiles $s_1,s_2,s_3\in T$ portraying \Cor all having distinct north and east glues,
	\item four distinct supertiles $s_1,s_2,s_3,s_4\in T$ portraying \Cor all having distinct north glues and together having at least two distinct east glues,
	\item four distinct supertiles $s_1,s_2,s_3,s_4\in T$ portraying \Cor all having distinct east glues and together having at least two north glues, or
	\item eight distinct supertiles $s_1,\ldots,s_8\in T$ portraying \Cor all having distinct east or north glues.
\end{compactenum}
Indeed, if none of these conditions were true for $Q_\Theta$, then, using the construction given in Lemma~\ref{lem:st:to:t}, we could generate a \tas $T$ which self-assembled $P$ and invalidated Lemma~\ref{lem:p_or}.

By Lemma~\ref{lem:white}, two distinct supertiles portraying \Cor cannot share tile types in positions \tl A, they can only share tile types in positions \tl{B_1^*} and \tl{C_1} if their east outputs equal, and they can only share tile types in positions \tl{B_1^*}, \tl{C_1} if their north outputs equal.
For case~{\it i.)}\ we need at least 15 white tile types of which three can be used as control tiles; the left over tile type might be used as another control tile.
For cases~{\it ii.)}\ and~{\it iii.)}\ we need at least 16 white tile types of which four can be used as control tiles.
Case~{\it iv.)} is not possible because we would need at least 26 white tile types.
A more involved analysis reveals that only case~{\it ii.)}\ is possible, but for our purpose it is enough that the remaining 16 white tile types contain at most four types that can be used as control tiles.

The number of distinct supertiles in $Q_\Theta$ is limited by the number of tile types that can be used as control tiles plus the number incomplete supertiles.
We obtain that $Q_\Theta$ contains $k + 3 = m$ distinct supertiles as desired.
Using Lemma~\ref{lem:st:to:t}, the pattern $P$ can be self-assembled by a \tas $T$ with $m$ tile types.
\end{proof}

\section*{Conclusions}

We prove that \kmbpats, a natural variant of \kpats, is \NP-complete for $k = 3$.
Furthermore, we present a novel proof for the \NP-completeness of \pats and our proof is more concise than previous proofs.
We introduce several new techniques for pattern design in our proofs, in particular in Sect.~\ref{sec:bwg}, and we anticipate that these techniques can ultimately be used to prove that $2$-\mbpats and also $2$-\pats are \NP-hard.


\end{document}